\documentclass[11pt,reqno]{amsart}

\usepackage{amsmath,amssymb,amsthm,mathrsfs,eucal,textcomp}                                                        
\usepackage[utf8]{inputenc}
\usepackage[T1]{fontenc}
\usepackage{lmodern}
\usepackage{amsmath}
\usepackage{amsfonts}
\usepackage{amssymb}
\usepackage{caption}
\usepackage{mathrsfs}
\usepackage{amsthm}
\usepackage{tikz}
\usepackage{bm}
\usepackage{color}
\usepackage{rotating}
\usepackage{array}

\usetikzlibrary{calc,arrows}

\theoremstyle{plain}
\newtheorem{theorem}{Theorem}[section]
\newtheorem{lemma}[theorem]{Lemma}

\newtheorem{statement}{Statement}
\theoremstyle{definition}
\newtheorem{definition}[theorem]{Definition}
\newtheorem{remark}[theorem]{Remark}

\def\R{\mathbb{R}}
\def\Z{\mathbb{Z}}
\def\e{\varepsilon}
\def\J{\mathcal{J}}
\def\Jalt{\mathscr J}

\def\L{\mathcal{L}}

\def\E{\mathbb{E}}
\def\udf{\underline{f}}
\def\var{\text{Var}}
\def\pmuj{\text{Jmp}}
\def\H{\mathrm H}
\def\st{\mathsf t}
\def\sx{\mathsf x}
\def\ss{\mathsf s}

\usepackage[scale=0.775]{geometry}

\title[Quadratic and RI limits for LDF]{Quadratic and rate-independent limits for a Large-deviations functional}

\author{G. A. Bonaschi}
\address{Giovanni A. Bonaschi \\ Institute for Complex Molecular Systems and Department of Mathematics and Computer Science, Technische Universiteit Eindhoven, P.O. Box 513, 5600 MB, Eindhoven, The Netherlands \&  Dipartimento di Matematica, Universit\`{a} di Pavia, 27100 Pavia, Italy} 
\email{g.a.bonaschi@tue.nl}

\author{M. A. Peletier}
\address{Mark A. Peletier\\ Institute for Complex Molecular Systems and Department of Mathematics and Computer Science, Technische Universiteit Eindhoven, Den Dolech 2, P.O. Box 513, 5600 MB Eindhoven, The Netherlands}
\email{m.a.peletier@tue.nl}

\date{ \today}

\begin{document}

\begin{abstract}
We construct a stochastic model showing the relationship between noise, gradient flows and rate-independent systems. The model consists of a one-dimensional birth-death process on a lattice, with rates derived from Kramers' law as an approximation of a Brownian motion on a wiggly energy landscape. Taking various limits we show how to obtain a whole family of generalized gradient flows, ranging from quadratic to rate-independent ones, connected via `$L \log L$' gradient flows. This is achieved via Mosco-convergence of the renormalized large-deviations rate functional of the stochastic process.
\end{abstract}


\maketitle


\section{Introduction}

\subsection{Variational evolution}
Two of the most studied types of variational evolution,  `Gradient-flow evolution' and `Rate-independent evolution', differ in quite a few aspects. Although both are driven by the variation in space and time of an energy, gradient flows are in fact driven by energy gradients, while in practice rate-independent systems are driven by changes in the external loading (represented by the time variation of the energy). As a result, gradient-flow systems have an intrinsic time scale, while rate-independent systems (as the name signals) do not, and the mathematical definitions of solutions of the two are rather different~\cite{Mielke05a,AGS}.

Despite this they share a common structure. Both can be written, at least formally, as
\begin{equation}
\label{eq:introductionGF}
0\in \partial\psi(\dot x(t)) + \mathrm{D}_x E(x(t),t).
\end{equation}
Here $E$ is the energy that drives the system, and the convex function $\psi$ is a \emph{dissipation potential}, with subdifferental $\partial \psi$. For gradient flows, $\psi$ typically is quadratic, and $\partial\psi$ single-valued and linear; for rate-independent systems, $\psi$ is $1$-homogeneous, and $\partial\psi$ is a degenerate monotone graph. 

Rate-independent systems have some unusual properties. Solutions are expected to be discontinuous, and therefore the concept of smooth solutions is meaningless. Currently two rigorous definitions of weak solutions are used, which we refer to as `energetic solutions'~\cite{MainikMielke05} and `$BV$ solutions'~\cite{MielkeRossiSavare12a}. Heuristically, the first corresponds to the principle `jump whenever it lowers the energy', while the second can be characterized as `don't jump until you become unstable'. For time-dependent \emph{convex} energies the two definitions coincide, but in the non-convex case they need not be.

Various rigorous justifications of rate-independent evolutions have been constructed, which underpin the rate-independent nature by obtaining it through upscaling from a  `microscopic' underlying system (e.g.~\cite{AbeyaratneChuJames96,Cagnetti08,Dal-MasoDeSimoneMora06,Dal-MasoDeSimoneMoraMorini08,Fiaschi09a,Mielke12,PuglisiTruskinovsky05,Sullivan09TH}).
The general approach in these results is to choose a microscopic model with a component of gradient-flow type (quadratic dissipation, deemed more `natural'),  and then take a limit which  induces the vanishing of the quadratic behaviour and the appearance of the rate-independent behaviour.

While these results give a convincing explanation of the rate-independent nature, they all are based on \emph{deterministic} microscopic models. Other arguments suggest that  the rate-independence may arise through the interplay between thermal noise and a rough energy landscape. A well-studied example of this is the non-trivial temperature-dependence of the yield stress in metals, which shows that the process is thermally driven (e.g.~\cite{Basinski59}), together with many classical non-rigorous derivations of rate-independent behaviour~\cite{Becker25,Orowan40,KrauszEyring75}.

Recently, stochasticity has also been shown to play a role in understanding the origin of various gradient-flow systems, such as those with Wasserstein-type metrics~\cite{AdamsDirrPeletierZimmer11,DirrLaschosZimmer12,AdamsDirrPeletierZimmer13,Renger13TH,MielkePeletierRenger13TR}. In this paper we ask the question whether these different roles of noise can be related:
\begin{quote}
\emph{What is the relationship between noise, gradient flows, and rate-independent systems?}
\end{quote}

We will provide a partial answer to this question by studying a simple stochastic model below. By taking various limits in this model,  we obtain a full continuum of behaviours, among which rate-independence and quadratic gradient flow can be considered  extreme cases. In this sense both rate-independent and quadratic dissipation arise naturally from the same stochastic model in different limits.

\subsection{The model}
The model of this paper is a continuous-time Markov jump process $t\mapsto X^n_t$ on a one-dimensional lattice,  sketched in Figure~\ref{fig:birth-death}. Denoting by $1/n$ the lattice spacing, we will be interested in the continuum limit as $n\to\infty$. 

The evolution of the process can be described as follows. Assume that a smooth function $(x,t) \mapsto E(x,t)$ is given and fix the origin as initial point. If the process is at the position $x$ at time $t$, then it jumps in continuous time to its neighbours $(x-1/n)$ and $(x+1/n)$ with rate $nr^-$ and $nr^+$, where $r^{\pm}(x)= \alpha \exp( \mp \beta \nabla E(x,t) ) $ (throughout we use $\nabla E(x,t)$ for the derivative with respect to $x$).
\begin{figure}[h]
\centering
\begin{tikzpicture}[scale=1.4]
\draw [semithick, color=black] (-4,0) -- (3,0);
\draw [semithick, color=black] (-3,.3) -- (-2,.3);
\draw [semithick, color=black] (-3,.4) -- (-3,.2);
\draw [semithick, color=black] (-2,.4) -- (-2,.2);
\draw (-1,-.1) node[anchor=north] {$x - \frac1n$};
\draw (0,-.1) node[anchor=north] {$x$};
\draw (1,-.1) node[anchor=north] {$x+\frac1n$};
\draw (-2.5,.3) node[anchor=south] {$\frac1n$};
\foreach \i in {-3,...,2}{ 
\draw [semithick, color=black] (\i,-.1) -- (\i,.1) ;}
\draw [->, semithick, color=black] (0,.2) to [out=120,in=60] (-1,.2);
\draw [->, semithick, color=black] (0,.2) to [out=60,in=120] (1,.2);
\draw (.5,1) node[anchor=north] {$nr^+$};
\draw (-.5,1) node[anchor=north] {$nr^-$};
\draw (5,.5) node[anchor=center] {$r^{\pm}(x)= \alpha e^{ \mp \beta \nabla E(x,t) } $};
\draw (5,-.5) node[anchor=center] {$x \in \frac1n \Z$};
\end{tikzpicture}
\caption{The one-dimensional lattice with spacing $1/n$. The jump rates~$r^+$ and~$r^-$ depend on two parameters $\alpha$ and~$\beta$ and on the derivative of the function~$E$.}
\label{fig:birth-death}
\end{figure}
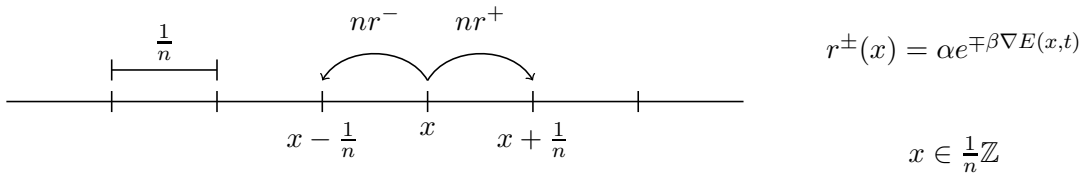

The choice of this stochastic process is inspired by the noisy evolution of a particle in a wiggly energy landscape. An example could be that of a Brownian particle in an energy landscape of the form $\mathcal E_n(x,t) = E(x,t) + n^{-1} \mathrm{e}(nx)$, where $E$ is the smooth energy introduced above, and $\mathrm e$ is a fixed periodic function.  If the noise is small with respect to the variation of $\mathrm{e}$ $(\max \mathrm{e} - \min \mathrm{e})$, then this Brownian particle will spend most of its time near the wells of $\mathcal E_n$, which are close to the wells of $\mathrm{e}(n\,\cdot\,)$. Kramers' formula~\cite{Kramers40,Berglund11TR} provides an estimate of the rate at which the particle jumps from one well to the next; in Section~\ref{subsec:modelling} below we show how some approximations lead to the jump rates~$r^\pm$ above.  

The jump process of Figure~\ref{fig:birth-death} has a bias in the direction $-\nabla E$ of magnitude 
\begin{equation}
\label{eq:diff-of-rates}
r^+ - r^- = -2\alpha  \sinh (\beta \nabla E),
\end{equation}
and we will see this expression return as a drift term in the limit problem. The parameter $\alpha$ characterizes the rate of jumps, and thus fixes the global time scale of the process; the parameter $\beta$ should be thought as the inverse of temperature, and characterizes the size of the noise.

\subsection{Heuristics}
We now give a heuristic view of the dependence of this stochastic process on the parameters $n$, $\alpha$, and $\beta$, and in doing so we look ahead at the rigorous results that we prove below. 

First, as $n\to\infty$, the process $X^n$ becomes deterministic, as might be expected, and its limit $x$ satisfies the differential equation suggested by~\eqref{eq:diff-of-rates}:
\begin{equation}
\label{eq:xdotfirst}
\dot x(t) = -2\alpha  \sinh (\beta \nabla E(x(t),t)).
\end{equation}
Equation \eqref{eq:xdotfirst} is of the form \eqref{eq:introductionGF} with $(\partial \psi)^{-1}=2\alpha \sinh (\beta \cdot)$. From the viewpoint of the gradient-flow-versus-rate-independence discussion above, the salient feature of the function  $\xi\mapsto 2\alpha\sinh(\beta \xi)$ is that it embodies both quadratic and rate-independent behaviour in one and the same function, in the form of  limiting behaviours according to the values of the parameters $\alpha,\beta$. This is illustrated by Figure~\ref{fig:visc-norm-RI}, as follows. On one hand, if we construct a limit by zooming in to the origin, corresponding to $\beta \to 0$, $\alpha \sim w/\beta$ (the left-hand figure), then we find a limit that is linear; on the other hand, if we zoom out, and rescale with $\beta \to \infty$ and $\alpha \sim e^{-\beta A}$, then the exponential growth causes the limit to be the monotone graph in the right-hand side. 

These two limiting cases correspond to a gradient-flow and a rate-independent behaviour respectively. In formulas, as $\alpha\to\infty$ and $\beta\to0$ with $\alpha\beta\to\omega$ for fixed $\omega>0$, then equation~\eqref{eq:xdotfirst} converges to 
\begin{equation}
\label{eq:Q}
\dot x(t) = -2\omega \nabla E(x(t),t),
\end{equation}
which is a gradient flow of $E$. The limit $\alpha\to\infty$ corresponds to large rate of jumps in the underlying stochastic process, while $\beta\to0$ corresponds to a weak influence of the energy gradient. 

In the other case, as $\alpha\to0 $ and $\beta\to\infty$ with $\alpha\sim \exp(-\beta A)$, the rate of jumps is low, but the influence of the energy becomes large. Formally, we find the limiting equation
\begin{equation}
\label{eq:RI}
\dot x(t) \in m_A(-\nabla E(x(t),t)), 
\qquad\text{where}\qquad
m_A(\xi) = \begin{cases}
\emptyset  & \text{if } \xi < -A,\\
[-\infty,0] & \text{if } \xi = -A,\\
\{0\} & \text{if } -A < \xi < A,\\
[0,\infty] & \text{if } \xi = A,\\
\emptyset & \text{if } \xi>A.
\end{cases} 
\end{equation}
Again formally, in this limit the system can only move while $\nabla E=\pm A$; whenever the force $|\nabla E|$ is less than $A$, the system is frozen, while values of $|\nabla E|$ larger than $A$ should never appear. In Section~\ref{sec:rilimit} we obtain a rigorous version of this evolution as the limit system. 

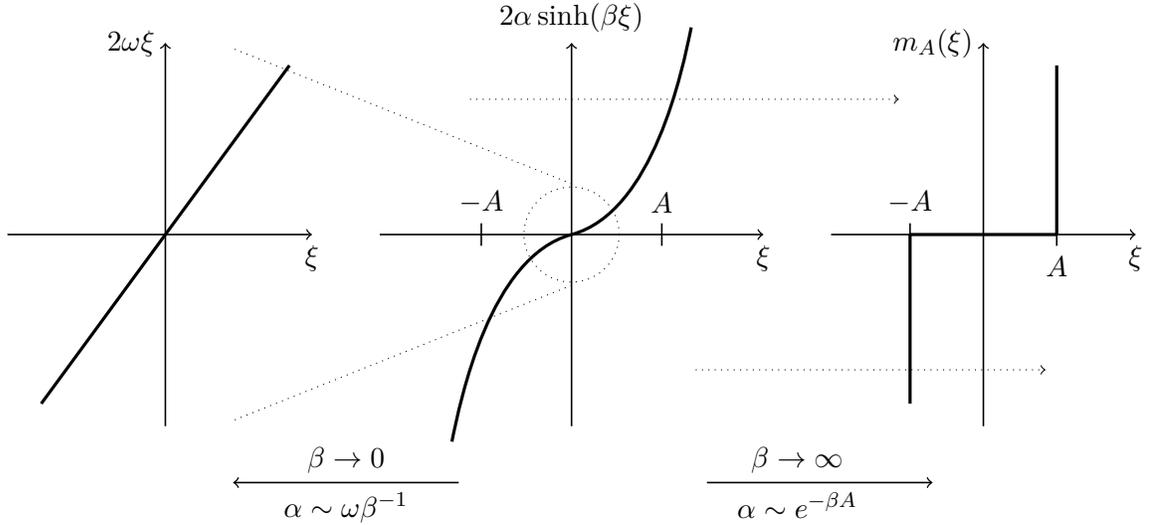
\begin{figure}[t]
\centering
\begin{tikzpicture}[
        scale=1.5,
        axis/.style={semithick, ->, >=stealth'},
        important line/.style={thick},
        dashed line/.style={dashed, thick},
        every node/.style={color=black,}
     ]
    \coordinate (beg_1) at (0,-.5);
    \coordinate (beg_2) at ($(beg_1)+(1,0)$);
    \coordinate (dev_1) at ($(beg_2)+(0,-.75)$);
    \coordinate (xint) at (3,0);
    \coordinate (end) at (5,1.25);

    \draw[->, semithick] (.3,0)  -- (3.7,0) node(xline)[below] {$\xi$};
    \draw[->, semithick] (2,-1.7) -- (2,1.7) node(yline)[above] {$2\alpha\sinh(\beta\xi)$};
    \draw[very thick, rotate around={-165:(2,0)}]
        (2,0) parabola  (3.5,1.5);
    \draw[very thick, rotate around={15:(2,0)}]
        (2,0) parabola  (3.5,1.5);
    \draw[semithick] (2.8,-.1) -- (2.8,.1) node[above] {$A$};
    \draw[semithick] (1.2,-.1) -- (1.2,.1) node[above] {$-A$};

    \draw[dotted] (2,0) circle (12pt);
    \draw[dotted] (2,.45) -- (-1,1.65);
    \draw[dotted] (2,-.45) -- (-1,-1.65);

    \draw[<-, semithick] (-1,-2.2) -- (1,-2.2);
    \draw (0,-2.2) node[anchor=south] {$\beta \to 0$};
    \draw (0,-2.2) node[anchor=north] {$\alpha \sim \omega\beta^{-1}$};

    \draw[->, dotted] (1.1,1.2) -- (4.9,1.2);    
    \draw[->, dotted] (3.1,-1.2) -- (6.2,-1.2);    

    \draw[->, semithick] (3.2,-2.2) -- (5.2,-2.2);    
    \draw (4,-2.2) node[anchor=south] {$\beta \to \infty$};
    \draw (4,-2.2) node[anchor=north] {$\alpha \sim e^{-\beta A}$};

   \draw[->, semithick] (-3,0) -- (-.3,0) node(xline)[below] {$\xi$};
   \draw[->, semithick] (-1.6,-1.7) -- (-1.6,1.7)  node(yline)[left] {$2\omega\xi $}; 
   \draw[very thick] (-2.7,-1.5) -- (-.5,1.5);

   \draw[->, semithick] (4.3,0) -- (7,0) node(xline)[below] {$\xi$};   
   \draw[->, semithick] (5.65,-1.7) -- (5.65,1.7)  node(yline)[left] {$m_A(\xi) $}; 
   \draw[very thick] (5,0) -- (6.3,0); 
   \draw[semithick] (5,-.1) -- (5,.1) node[anchor=south] {$-A$};
   \draw[semithick] (6.3,.1) -- (6.3,-.1) node[anchor=north] {$A$};
   \draw[very thick] (5,0) -- (5,-1.5);
   \draw[very thick] (6.3,0) -- (6.3,1.5);
    
\end{tikzpicture}
\caption{The middle graph shows the function $\xi\mapsto 2\alpha \sinh(\beta\xi)$ for moderate values of $\alpha$ and $\beta$. The left graph shows the limit for $\beta \to 0$, similar to zooming in to the region close to the origin; this limit is linear. The figure on the right shows the limiting behaviour when $\beta \to \infty$, for a specific scaling of $\alpha$. This second limit does not exist as a function, but only as a graph (a subset of the plane) defined in \eqref{eq:RI}.}
\label{fig:visc-norm-RI}
\end{figure}

\subsection{Large deviations, gradient flows, and variational formulations}
\label{subsec:ldpgf-intro}

Before we describe the results of this paper, we comment on the methods that we use. We previously introduced the concept of gradient flows and now we introduce the one of \emph{large deviations} (both are defined precisely in Section~\ref{sec:GF&LD}).  

In the context of stochastic systems, the theory of large deviations provides a characterization of the probability of rare events, as some parameter---in our case $n$---tends to infinity. In the case of stochastic processes, this leads to a \emph{large-deviations rate function} $\Jalt$ that is defined on a suitable space of curves. It is now known that many gradient flows and large-deviations principles are strongly connected~\cite{AdamsDirrPeletierZimmer11,DuongLaschosRenger13,AdamsDirrPeletierZimmer13,MielkePeletierRenger13TR}. In abstract terms, the rate function $\Jalt$ of the large-deviations principle simultaneously figures as the defining quantity of the gradient flow, in the sense that 
\begin{align*}
&\Jalt\geq 0;\\
&\text{$t\mapsto z(t)$ is a solution of the gradient flow }\Longleftrightarrow \Jalt(z(\cdot)) = 0.
\end{align*}
The components of the gradient flow (the energy $E$ and dissipation potential $\psi$) can be recognized in~$\Jalt$.
In~\cite{MielkePeletierRenger13TR} it was shown how jump processes may generate large-deviations rate functions with non-quadratic dissipation, leading to the concept of \emph{generalized} gradient flows (see also~\cite{MielkeRossiSavare09,MRS12,DuongPeletierZimmer13}). 

The central tool in this paper is this functional $\Jalt$ that characterizes both the large deviations of the stochastic process and the generalized gradient-flow structure of the limit. Our convergence proofs will be stated and proved using only this functional, giving a high level of coherence to the results. 

\subsection{Results}

\label{sec:mainresults}

\begin{figure}[t]
\begin{tabular}{c}
\begin{tikzpicture}[scale=2]
\draw (0,0) node[anchor=center] {$X^n$};
\draw [->, semithick, color=black] (0.3,0) -- (2.1,0);
\draw (1.3,0) node[anchor=south] {$\beta = (hn)^{-1}$};
\draw (1.3,0) node[anchor=north] {$n \to \infty$};
\draw (3,.1) node[anchor=center] {$dY^{h}_t= - 2\omega \, \nabla E \, dt$};
\draw (3,-.2) node[anchor=center] {$+ \, \sqrt{2 \omega h} \, dW$};
\draw [->, dashed, color=black] (0,-.3) -- (0,-2.7);
\draw (0,-1.5) node[anchor=south,rotate=-90] {$n \to \infty$};
\draw (0,-1.5) node[anchor=north,rotate=-90] {$\alpha, \beta \;$ fixed};
\draw (0,-3) node[anchor=center] {$\Jalt_{\alpha,\beta}$};
\draw [->, semithick, color=black] (.3,-.3) -- (2.7,-2.7);
\draw (1.5,-1.5) node[anchor=north,rotate=-45] {$\beta > n^{-1}$};
\draw (1.5,-1.5) node[anchor=south,rotate=-45] {$n\to \infty$};
\draw [->, semithick, color=black] (.3,-3) -- (2.7,-3);
\draw (1.5,-3) node[anchor=south] {$\beta \to 0$};
\draw (1.5,-3) node[anchor=north] {$\alpha \to \infty$};
\draw (3,-3) node[anchor=center] {$\J_Q$};
\draw [->, dashed, color=black] (3,-.5) -- (3,-2.7);
\draw (3,-1.5) node[anchor=south,rotate=-90] {$h\to 0$};
\draw [->, dotted, color=black] (-.3,-.3) -- (-2.7,-2.7);
\draw (-1.5,-1.5) node[anchor=south,rotate=45] {?};
\draw [->, semithick, color=black] (-.3,-3) -- (-2.7,-3);
\draw (-1.5,-3) node[anchor=south] {$ \beta \to \infty$};
\draw (-1.5,-3) node[anchor=north] {$\alpha \to 0$};
\draw (-3,-3) node[anchor=center] {$\J_{RI}$};
\end{tikzpicture}
\end{tabular}
\caption{The figure is a schematic representation of this paper. In the top center there is the generator of the Markov process. The arrows starting from $X^n$ represent the limiting behaviour for $n \to \infty$ in different regimes. The center arrow represents the limit with $\alpha$ and $\beta$ fixed and ends at the rate functional $\Jalt_{\alpha,\beta}$ (Statement~\ref{statement:LD}). Statement~\ref{statement:Commutation} is represented by the right side of the Figure. The limits with $\beta=\beta(n)$  show that the limiting behaviour may be either a Brownian motion with drift (B2a) or a gradient flow characterized by the rate functional $\J_Q$ (B1). In the bottom part there are the two Mosco-limits, representing Statement~\ref{statement:Connection} where (A1) is the right arrow and (A2) the left one. Dashed lines are known results, the thick lines are our contribution, and the dotted line is an open problem.}
\label{fig:ideaofpaper}
\end{figure}
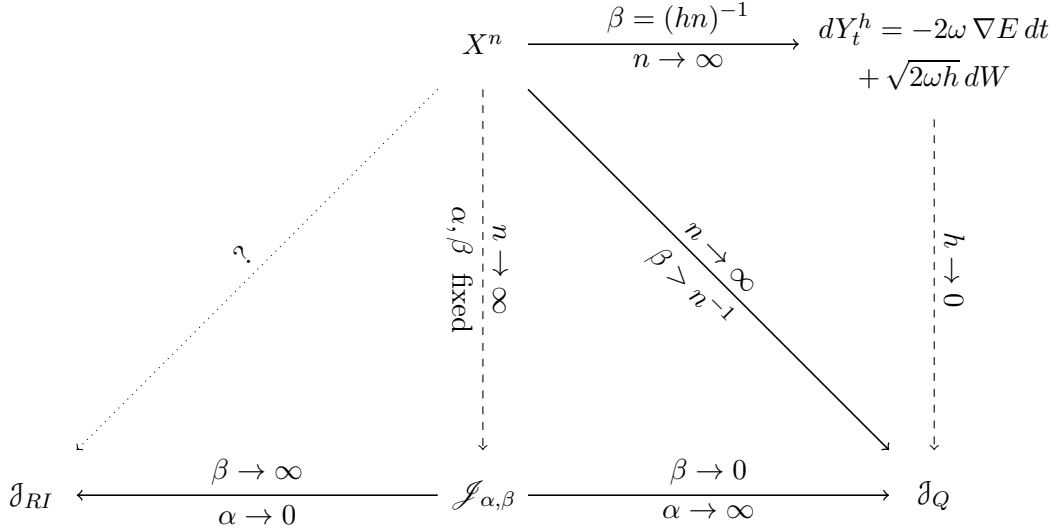

In this section we give a non-rigorous description of the results of this paper, with pointers towards the rigorous theorems later in the paper. 

Fix an energy $E \in C^1(\R\times[0,T])$.
We start with the large deviations result for the jump process $X^n$ due to Wentzell~\cite{Wen77,Wen90}.

\begin{statement}[Large deviations]
\label{statement:LD}
For constant $\alpha$ and $\beta$, $X^n$ satisfies a large-deviations principle for $n \to \infty$, with rate function $\Jalt_{\alpha,\beta}$ given in~\eqref{def:lardevfun}. Moreover, the minimizer of $\Jalt_{\alpha,\beta}$ satisfies the generalized gradient flow equation~\eqref{eq:xdotfirst}.
\end{statement}

This result is stated in Theorem~\ref{thm:LDF}  with a sketch of the proof as it is presented in the introduction of~\cite{FK06}. In accordance with the discussion above, $\Jalt_{\alpha,\beta}(x)\geq 0$ for all curves $x:[0,T]\to\R$, and $\Jalt_{\alpha,\beta}(x)=0$ if and only if $x$ is a solution of~\eqref{eq:xdotfirst}. 

\medskip

Next we prove that the functional $\Jalt_{\alpha,\beta}$ converges to two functionals $\J_Q$ and to $\J_{RI}$ in the sense of Mosco-convergence, defined in \eqref{def:moscoconvergence}, when $\alpha$ and $\beta$ have the limiting behaviour of Figure~\ref{fig:visc-norm-RI}. The limiting functionals drive respectively a quadratic gradient-flow and a rate-independent evolution. 

\begin{statement}[Connection]
\label{statement:Connection}
\begin{itemize}
\item[A1.] Let $\alpha\to\infty$ and $\beta\to 0$, such that $\alpha\beta\to\omega$, for some $\omega>0$ fixed. Then we have that, after rescaling, $\Jalt_{\alpha,\beta} \to \J_Q$, where $\J_Q$ is defined in~\eqref{eq:defJ0}; moreover $\J_Q(x) =0$ if and only if  $x$ solves~\eqref{eq:Q}.
\item[A2.] Let $\beta \to \infty$ and choose $ \alpha =e^{-\beta A}$, for some $A>0$ fixed. Then, after an appropriate rescaling, $\Jalt_{\alpha,\beta}\to \J_{RI}$, where $\J_{RI}$ is given in~\eqref{def:JRI}; moreover $\J_{RI}(x) = 0$ if and only if $x$ is an appropriately defined solution of~\eqref{eq:RI}.
\end{itemize}
\end{statement}

The Mosco-convergence stated above also implies that minimizers $x_{\alpha,\beta}$ of $\Jalt_{\alpha,\beta}$ converge to the minimizers of $\J_Q$ and $\J_{RI}$ in the corresponding cases, i.e.\ that the solutions $x_{\alpha,\beta}$ of~\eqref{eq:xdotfirst} converge to the solutions of~\eqref{eq:Q} and~\eqref{eq:RI}. Point~A1 is proven in Theorem~\ref{thm:quadlimit}, and point~A2 in Theorem~\ref{thm:RIgammaconv}. 

\medskip
Together, Statements~\ref{statement:LD} and~\ref{statement:Connection} describe a sequential limit process: first we let $n\to\infty$, and then we take limits in $\alpha$ and $\beta$. 
For the quadratic case (A1) we can also combine the limits:

\begin{statement}[Combining the limits]
\label{statement:Commutation}
Let $n\to\infty$, and take $\beta=\beta_n \sim n^{-\delta}$ for some $0\leq \delta\leq 1$; let $\alpha = \alpha_n$ be such that $\alpha_n\beta_n \to\omega$, for some fixed $\omega>0$. \begin{itemize}
\item[B1.] First let $0 < \delta < 1$. Then $X^n$ satisfies a large-deviations principle as $n \to \infty$,  with rate function $\J_Q$; the Markov process $X^n$ has a deterministic limit \eqref{eq:Q}, and this limit minimizes $\J_Q$ (as we already mentioned). 
\item[B2a.] In the case  $\delta=1$, let $\alpha_n \beta_n \to \omega$, $n\beta_n\to 1/h$, for some $\omega,h>0$; then $X^n$ converges to the process $Y^h$ described by the SDE 
\begin{equation}
\label{eq:brownianplusdrif}
dY^h_t=-2\omega \nabla E(Y^h_t,t) \,dt + \sqrt{2\omega h}\;dW_t.
\end{equation}
\item[B2b.] The process $Y^h$ in \eqref{eq:brownianplusdrif} satisfies a large-deviations principle for $h \to 0$ with rate function $\J_Q$.
\item[B3.] The case $\delta=0$ corresponds to point A1, where first $n \to \infty$ and then $\beta \to 0$ and $\alpha \to \infty$.
\end{itemize}
\end{statement}

Point B1 of Statement~\ref{statement:Commutation} is given in Theorem~\ref{thm:lardevres}. Point B2a is given in Theorem~\ref{thm:convprocesses}; Point B2b is the well-known result of Freidlin and Wentzell~\cite[Ch.~4-Th.~1.1]{FW12}, and it is included in Theorem~\ref{thm:lardevres}. 

\begin{remark}
In this paper we consider only the one-dimensional case. We make this choice because the main goal of the paper is to show the connection and the interplay between large deviations and gradient flows, and the one-dimensionality allows us to avoid various technical complications. However, the generalization to higher dimension is in some cases just a change in the notation and does not require any relevant modification in some of the proofs.  The rate-independent limit in higher dimensions is non-trivial and it is the object of work in progress.
\end{remark}

%

\subsection{Modelling}
\label{subsec:modelling}

\begin{figure}[b]
\centering
\begin{tabular}{c}
\begin{tikzpicture}
 \draw[thick] (0,9.7) parabola bend (1,8) (2,8.5);
 \draw[thick] (2,8.5) parabola bend (3,9.2) (4,7.9); 
 \draw[thick] (4,7.9) parabola bend (5,7) (6,7.5);
 \draw[thick] (6,7.5) parabola bend (7,8.2) (8,7);
 \draw[thick] (8,7) parabola bend (9,6.2) (10,7);
 \draw[|-|](1,6.4) -- (5,6.4);
 \draw (3,6.4) node[anchor=north] {$2/n$};
 \draw[->] (-.7,5.7) -- (11,5.7) node[anchor=west] {$x$};
 \draw (0,9.7) node[anchor=south] {$\mathcal E_n(x,t)$};
 \draw[dotted] (3,8.7) -- (7,8.7);
 \draw[<->] (5,7.1) -- (5,8.6);
 \draw (5,7.85) node[anchor=west] {$\frac{1}{n} \Delta \mathrm{e} $};
 \draw[->] (3,8.75) -- (3,9.15);
 \draw (3.2,9.3) node[anchor=west] {$-\frac{1}{n} \nabla E $};
 \draw[<-] (7,8.25) -- (7,8.65);
 \draw (7,8.45) node[anchor=west] {$\frac{1}{n} \nabla E $};
\end{tikzpicture}
\end{tabular}
\caption{The global component $E$ perturbs the height of the energy barriers, leading to the formula for $r^\pm$ in Figure~\ref{fig:birth-death}.}
\label{fig:Kramers}
\end{figure}
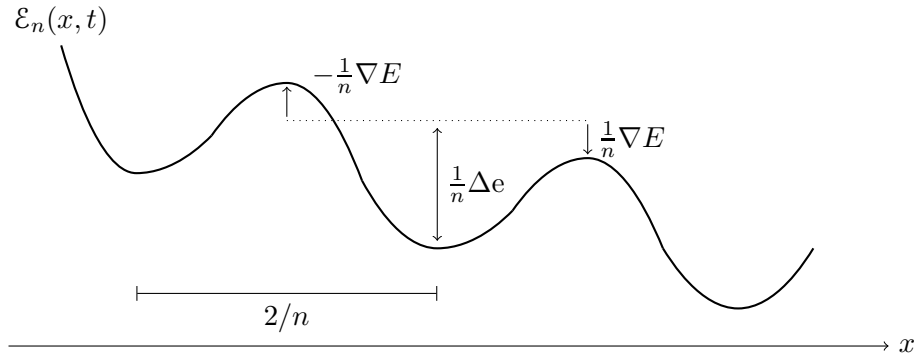

We mentioned above that the rates $r^\pm$ can be derived from Kramers' law; we now give some details. 

In the wiggly energy $\mathcal E_n(x,t) = E(x,t) + n^{-1}\mathrm{e}(nx)$, assume that $E$ varies slowly both in $x$ and $t$, $\mathrm{e}$ has period 2 and let $n$ be large. Then a small patch of the energy landscape of $\mathcal E_n$ looks like Figure~\ref{fig:Kramers}. The height of the energy barriers to the left and right of a well equals $\Delta \mathrm{e} := (\max \mathrm{e}-\min \mathrm{e})$, plus a perturbation from the smooth energy $E$, which to leading order  has size $n^{-1}\nabla E$. 

We now assume that the position $Z_t$ of the system solves the SDE
\[
dZ_t = -\nabla \mathcal E_n(Z_t,t) \, dt + \sqrt{\frac2{\beta n}} \, dW_t.
\]
Here $\beta$ characterizes the noise, and can be interpreted as $1/kT$ as usual, although in this case there is an additional scaling factor $n$. For sufficiently large $\beta$ the rates of escape from a well, to the left and to the right, are given by Kramers' law to be approximately
\begin{equation}
\label{def:heuristicrates}
a\exp \bigl[\beta(-\Delta \mathrm{e} \pm \nabla E)\bigr],
\end{equation}
where the minus sign applies to the rate of leaving to the right. Here $a$ is a constant depending on the form of $\mathrm{e}$~\cite{Kramers40,Berglund11TR}.
Writing 
\begin{equation}
\label{def:alpha}
\alpha = a\exp[-\beta \Delta \mathrm{e}],
\end{equation}
we find the rates $r^\pm$ of Figure~\ref{fig:birth-death}.

In this formula for $\alpha$, it appears that $\alpha$ and $\beta$ are coupled. From a modelling point of view, this is true: if one varies the temperature while keeping all other parameters fixed, then both $\beta$ and $\alpha$ will change. Note that the scaling regime $\alpha\sim e^{-\beta A}$ is exactly this case, with $A=\Delta \mathrm{e}$. On the other hand, the parameter $a$ in~\eqref{def:alpha} is still free, and this allows us to consider $\alpha$ and $\beta$ as independent parameters when necessary. 

\subsection{Outline}

The paper evolves as described in the following. In Section~\ref{sec:GF&LD} we introduce the  concepts of (generalized) gradient flows and of large deviations and we show the connection between the two concepts. In Section~\ref{sec:quadlim} we prove point~A1 of Statement~\ref{statement:Connection} and the whole Statement~\ref{statement:Commutation}. Then in Section~\ref{sec:rilimit} we introduce the space of functions of bounded variations, rate-independent systems, and we prove point~A2 of Statement~\ref{statement:Connection}. We end the paper in Section~\ref{sec:Discussion} with a final discussion.


\section{Gradient Flows \& Large Deviations}
\label{sec:GF&LD}


In the introduction we mentioned that the methods of this paper make use of a certain unity between gradient flows and large-deviations principles: the same functional $\Jalt$ that defines the gradient flow also appears as the rate function of a large-deviations principle. We now describe gradient flows, large deviations, and this functional $\Jalt$.

\subsection{Gradient flows}

Given a $C^1$ energy $E:\R \to \R$, we call a \emph{gradient flow} of $E$ the flow generated by the  equation
\begin{equation*}
\dot{x} = - \nabla E (x).
\end{equation*}
The energy $E$ decreases along a solution, since
\begin{equation*}
\frac{d}{dt}E(x(t))= \nabla E \cdot \dot{x} = - |\nabla E|^2 = - |\dot{x}|^2.
\end{equation*}
Adopting the notation $\psi(\xi)=\frac{\xi^2}{2}$ and $\psi^*(\eta)=\frac{\eta^2}{2}$, this identity can be integrated in time to find
\begin{equation*}
\int_0^T \left( \psi(\dot{x}) + \psi^*(-\nabla E) \right) dt + E(x(T)) - E(x(0)) = 0
\end{equation*}
In this paper we study a generalized concept of gradient flow, considering the energy equality for a broader class of couples $\psi,\psi^*$. We will allow the energy to be also dependent on time. We recall the definition of the Legendre transform: given $\psi: \R \to \R$ we define the transform $\psi^*$ as 
\begin{equation*}
\psi^* (w) := \sup_{v} \left\{ v \cdot w - \psi(v)\right\}.
\end{equation*}
In the following, apart from the rate-independent case, we will assume that $\psi \in C^1(\R)$ is symmetric, superlinear and convex, so that  its Legendre transform $\psi^*$ will share the same properties as well.

A curve $x:[0,T] \to \R$ is an absolutely continuous curve, i.e. $x \in AC(0,T)$, if for every $\e > 0$, there exists a $\delta > 0$ such that, for every finite sequence of pairwise disjoint intervals $(t_j,\tau_{j}) \subset [0,T]$ satisfying $\sum |t_j - \tau_{j}| < \delta$, then $\sum |x(t_j)-x(\tau_{j})| < \e$. The space $AC(0,T)$ coincides with the Sobolev space $W^{1,1}(0,T)$~\cite[Ch.~8]{Brezis11}.

\begin{definition}[Generalized gradient flow]
\label{def:gengraflo}
Given an energy $E \in C^1(\R \times [0,T])$, a convex dissipation potential $\psi \in C^1(\R)$ with $\psi(v)=\psi(-v)$, let $\psi^*$ be its Legendre transform. Then a curve $x \in AC(0,T)$ is a \emph{(generalized) gradient flow} of $E$ with dissipation potential $\psi$ in a given time interval $[0,T]$, if it satisfies the  energy identity
\begin{equation}
\label{eq:gen_ene_equ}
\int_0^T \bigg( \psi(\dot{x}(t)) + \psi^*\bigl(-\nabla E (x(t),t)\bigr) \bigg) dt + E(x(T),T) - E(x(0),0) - \int_0^T \partial_t E(x(t),t) dt = 0.
\end{equation}
\end{definition}

Note that the left-hand side of~\eqref{eq:gen_ene_equ} is non-negative for any function $x$, since
\[
\frac{d}{dt} E(x(t),t)  = \nabla E(x(t),t) \dot x(t) + \partial_t E(x(t),t) 
\geq -\psi(\dot x(t)) - \psi^*(-\nabla E(x(t),t)) + \partial_t E(x(t),t) .
\]
From this inequality one deduces that \emph{equality} in~\eqref{eq:gen_ene_equ}, as required by Definition~\ref{def:gengraflo}, implies that for almost all $t\in[0,T]$
\begin{equation}
\label{eq:GGF}
\dot{x}(t) = \partial \psi^* \bigl(-\nabla E(x(t),t)\bigr),
\end{equation}
where $\partial\psi^*$ is the subdifferential of $\psi^*$. We will not use this form of the equation; the arguments of this paper are based on Definition~\ref{def:gengraflo} instead. 

Existence and uniqueness of classical gradient-flow solutions in $\R$ (i.e.\ with quadratic~$\psi$) follows from classical ODE theory; in recent years the theory has been extended to metric spaces and spaces of probability measures~\cite{AGS}. 

In our case we will require the energy to satisfy the following conditions
\begin{equation}
\label{cond:energy}
\begin{cases}
E \in C^1(\R \times [0,T]), \,  E \geq 0, \\
|\nabla E | \leq R < \infty, \\
\nabla E \text{ is uniformly Lipschitz continuous in } t. \\
\end{cases}
\end{equation}

\begin{remark}
Definition~\ref{def:gengraflo} requires $\psi$ to be strictly convex. 
The rate-independent evolution~\eqref{eq:RI} is formally the case of a non-strictly convex, $1$-homogeneous dissipation potential $\psi$, and for this case there are several natural ways to define a rigorous solution concept. In Section~\ref{sec:rilimit} we show how generalized gradient flows for finite $\alpha$ and $\beta$, with strictly convex $\psi$,  converge to a \emph{specific} rigorous rate-independent solution concept, the so-called \emph{BV solutions}~\cite{MielkeRossiSavare09,MielkeRossiSavare12a}. We define this concept in Definition~\ref{def:RIgraflo}.
\end{remark}

Returning to the unity between gradient flows and large deviations, for generalized gradient flows the functional $\Jalt$ mentioned before is the left-hand side of~\eqref{eq:gen_ene_equ}. 


\subsection{Large deviations}

`Large deviations' of a random variable are rare events, and large-deviations theory characterizes the rarity of certain rare events for a sequence of random variables. 
Let $\left\{ X^n \right\}$ be such a sequence of random variables with values in some metric space.
\begin{definition}[\cite{Var66}]
\label{def:ldp}
$\{ X^n \}$ satisfies a large-deviations principle (LDP) with speed $a_n \to \infty$, if there exists a lower semicontinuous function $\Jalt:S\to [0,\infty]$ with compact sublevel sets such that for each open set $O$, 
\[
\liminf_{n \to \infty} \frac{1}{a_n}\log P(X^n \in O) \geq - \inf_{x \in O}\Jalt(x),
\]
and for each closed set $C$
\[
\limsup_{n \to \infty} \frac{1}{a_n}\log P(X^n \in C) \leq - \inf_{x \in C}\Jalt(x).
\]
The function $\Jalt$ is called the \emph{rate function} for the large-deviations principle.
\end{definition}

Intuitively, the two inequalities above state that 
\[
\text{Prob}(X^n \simeq x) \sim e^{-a_n \Jalt(x)},
\]
where we purposefully use the vague notations $\simeq$ and $\sim$; the rigorous versions of these symbols is exactly given by Definition~\ref{def:ldp}.

\begin{remark}
Typically, the rate function for Markov processes contains a term $I_0$ characterizing the large deviations of the initial state $X^n(0)$. In the following we will always assume that the starting point will be fixed, or at least that $X^n(0) \to x_0$, so that $I_0(x)$ equals $0$ if $x=x_0$, and $+\infty$ otherwise, so we will disregard $I_0$.
\end{remark}

\subsection{The Feng-Kurtz method. }
Feng and Kurtz created a general method to prove large-deviations principles for Markov processes~\cite{FK06}. The method provides both a formal method to calculate the rate functional and a rigorous framework to prove the large-deviations principle. Here we present only the formal calculation.

Consider a sequence of Markov processes $\{ X^n \}$ in $\R$, which we take time-invariant for the moment, and consider the corresponding evolution semigroups $\{ S_n(t) \}$ defined by 
\[
S_n(t)f(x)= \mathbb{E} \left[ f(X^n(t)) \, | \, X^n(0)=x \right], \qquad f \in C_b(\R),
\]
satisfying
\[
\frac{d}{dt}S_n(t)f=\Omega_n S_n(t)f, \qquad S_n(0)f=f,
\]
where $\Omega_n$ is the generator of $X_n$. For any time interval $[0,T]$, where $T$ may be infinite, $X^n(\cdot)$ is an element of the Skorokhod space $D([0,T])$, the space of \emph{cadlag} functions (right-continuous and bounded). 
To obtain the rate functional we  define the \emph{non-linear generator}
\[
\big( \H_nf \big)(x) := \frac{1}{a_n}e^{-a_nf(x)} \big( \Omega_n e^{a_nf} \big)(x).
\]
If $\H_n\to \H$ in some sense, and if $\H f$  depends locally on $\nabla f$, we  then define the \emph{Hamiltonian} $H(x,p)$ through
\[
\H f(x)=:H(x,\nabla f(x)).
\]
By computing the Legendre transform of $H(x,p)$ we obtain the  \emph{Lagrangian}
\[
L(x,\dot{x})=\sup_{p \in \R} \left\{ \dot{x} \cdot p - H(x,p) \right\}.
\]
The Feng-Kurtz method then states, formally, that $\{ X^n \}$ satisfies a large-deviations principle in $D([0,T])$ with speed $a_n$, with a rate function
\begin{equation*}
\Jalt(x)= 
\begin{cases}
\int_0^T L(x,\dot{x}) \,dt & \text{ if }x \in AC(0,T), \\
+ \infty & \text{otherwise}.
\end{cases}
\end{equation*}
In the book~\cite{FK06} a general method is described to make this algorithm rigorous.

\subsection{Large deviations of $X^n$}

We now apply this method to the process $X^n$ described in the introduction. It is a continuous time Markov chain, defined by its generator
\begin{equation}
\label{def:generator}
\Omega_n f (x):= n\alpha e^{-\beta \nabla E(x,t)}\left( f(x+\frac{1}{n})-f(x) \right) + n\alpha e^{\beta \nabla E(x,t)}\left( f(x-\frac{1}{n})-f(x) \right).
\end{equation}
For $f \in C_b(\R)$ the expected value $\E$ is defined as
\[
\E(f(X^n_t)|X^n_0)=\int_{\R}f(z)d\mu^n_t(z),
\]
with, denoting by $\Omega_n^T$ the adjoint of $\Omega_n$,
\begin{equation}
\label{eq:martingale}
\begin{cases}
\partial_t \mu_t=\Omega^T_n\mu_t, \\
\mu_0=\delta_{X^n_0}.
\end{cases}
\end{equation}
where $\mu^n_t$ is the law at time $t$ of the process $X^n$ started at the position $X^n_0$ at time $t=0$. Under the condition \eqref{cond:energy}, the martingale problem \eqref{eq:martingale} is well-posed, since the operator $\Omega_n$ is bounded in the uniform topology.



The rigorous proof of Statement~\ref{statement:LD} consists of the following theorem

\begin{theorem}[Large-deviations principle for $\Omega_n$]
\label{thm:LDF}
Let $E:\R\times[0,T] \to \R$ satisfying condition~\eqref{cond:energy}. Consider the sequence of Markov processes $\{ X^n \}$ with generator $\Omega_n$ defined in~\eqref{def:generator} and with $X^n(0)$ converging to $x_0$ for $n \to \infty$. Then the sequence $X^n$ satisfies a large-deviations principle in $D([0,T])$  with speed $n$ and rate function
\begin{equation}
\label{def:lardevfun}
\Jalt_{\alpha,\beta}(x):= \begin{cases}
\displaystyle \beta \int_0^T \left( \psi_{\alpha,\beta}(\dot{x}) + \psi^*_{\alpha,\beta}(\nabla E) + \dot{x} \nabla E \right) dt.
 & \text{for } x \in AC(0,T),\\
+\infty & \text{otherwise},
\end{cases}
\end{equation}
where 
\begin{equation}
\label{def:dissipation} 
\psi_{\alpha,\beta}(v):=\frac{v}{\beta} \log \left( \frac{ v + \sqrt{ v^2 + 4\alpha^2}}{2 \alpha}\right) - \frac{1}{\beta}\sqrt{v^2 + 4\alpha^2} + \frac{2 \alpha}{\beta}
\qquad\text{and}\qquad
\psi^*_{\alpha,\beta}(w)= \frac{2 \alpha}{\beta}\left( \cosh(\beta w) - 1 \right).
\end{equation}

\end{theorem}

The proof can be found in  \cite[Ch.~5-Th.~2.1]{FW12} when the energy $E$ is independent of time. In the general case of a time-dependent energy, the proof follows considering a space-time process, as shown in the proof of Theorem~\ref{thm:lardevres}.

Note that solutions of $\Jalt_{\alpha,\beta}(x) = 0$ satisfy the gradient-flow equation~\eqref{eq:GGF}, which in this case indeed is equation~\eqref{eq:xdotfirst}, i.e.
\[
\dot{x}=-2\alpha \sinh (\beta \nabla E(x,t)).
\]

\medskip
In the remainder of this paper we will consider sequences in $\alpha$ and $\beta$; to reduce notation we will drop the double index, writing $\psi_\beta$ and $\psi^*_{\beta}$ for $\psi_{\alpha,\beta}$ and $\psi^*_{\alpha,\beta}$; similarly we define the rescaled  functional~$\J_{\beta}$,
\begin{equation}
\label{def:jbetanotalpha}
\J_\beta(x) := \frac{1}{\beta}\Jalt_{\alpha,\beta}(x)= \begin{cases}
\displaystyle \int_0^T \bigg( \psi_{\beta}(\dot{x}(t)) + \psi_{\beta}^*(-\nabla E(x(t),t)) + \dot{x}(t) \nabla E(x(t),t) \bigg) \,dt & \text{for } x \in AC(0,T),\\
+\infty & \text{otherwise}.
\end{cases}
\end{equation}

\begin{remark}
The theorem above shows how the rate function $\Jalt_{\alpha,\beta}$ can be interpreted as defining a generalized gradient flow. This illustrates the structure of the fairly widespread connection between gradient flows and large deviations: in many systems the rate function not only defines the gradient-flow evolution, through its zero set, but the components of the gradient flow ($E$ and $\psi$) can be recognized in the rate function. This connection is explored more generally in~\cite{MielkePeletierRenger13TR} as we describe in the following. Define a so-called $L$-function $\mathcal{L}(z,s)$, positive, convex in $s$ for all $z$ and inducing an evolution equation. The authors of~\cite[Lemma~2.1 and Prop.~2.2]{MielkePeletierRenger13TR} show that if $ D_s \mathcal{L}(z,0)$ is an exact differential, say $DS(z)$, then it is possible to write $\mathcal{L}$ as
\begin{equation}
\label{def:MPRgradflow}
\mathcal{L}(z,s)=\Psi(z,s)+\Psi^*(z,-DS(z)) + \langle DS(z),s	 \rangle,
\end{equation}
where $\Psi^*$ can be expressed in terms of the Legendre transform $\mathcal{H}(z,\xi)$   of $\mathcal{L}$ as
\[
\Psi^*(z,\xi):=\mathcal{H}(z,DS(z)+\xi)-\mathcal{H}(z,DS(z)).
\] 

Applying the same procedure to our case, with $\mathcal{L}=L$ defined in \eqref{def:Lfunctionlagrangian}, we obtain after some calculations that
\[
S(z):=\frac{1}{2}\log\left( \frac{r^-(z)}{r^+(z)} \right).
\]
Substituting into $S(z)$ our choice for $r^+$ and $r^-$
\[
r^+=\alpha e^{-\beta \nabla E(z)}, \qquad r^-=\alpha e^{\beta \nabla E(x)},
\]
it follows  that $S(z)=\beta E(z)$.  
\end{remark}

\subsection{Calculating the large-deviations rate functional for \eqref{def:generator} }
\label{subsec:calc}
We conclude this section by calculating the rate function for the simpler situation when the jump rates $r^\pm$ are constant in space and time, as it is shown in the introduction of \cite{FK06}. This  formally proves Theorem~\ref{thm:LDF}, substituting in the end the expression or $r^{\pm}$ from \eqref{def:generator}. 

With constant jump rates, the generator reduces to 
\[
\Omega_n f(x)=nr^+\left[ f \left( x+\frac{1}{n} \right) - f(x) \right] + nr^-\left[ f \left( x-\frac{1}{n} \right)-f(x) \right],
\]
and for $n \to \infty$ it converges to $\Omega f(x)=(r^+-r^-)  \nabla f(x)$. As we said in the introduction, the process $X^n$ has a deterministic limit, i.e. $X^n \to x$ a.s., with $\dot{x}=r^+-r^-$. 

In order to calculate the rate functional, we compute the non-linear generator and the limiting Hamiltonian and Lagrangian. We have 
\[
\H_nf(x) = r^+\left[ e^{n(f ( x+1/n) - f(x))}-1 \right] + r^-\left[ e^{n(f ( x-1/n) - f(x))}-1\right],
\]
so that
\[
\lim_{n \to \infty} \H_n = H(x,p)=r^+ \left( e^p - 1 \right) + r^- \left( e^{-p} - 1 \right).
\]
We then obtain by an explicit calculation the Lagrangian
\begin{equation}
\label{def:Lfunctionlagrangian}
L(x,\dot{x})=\dot{x}\log \left( \frac{\dot{x}+\sqrt{\dot{x}^2+4r^+r^-}}{2r^+}\right)-\sqrt{\dot{x}^2+4r^+r^-}+r^+ + r^-,
\end{equation}
and substituting $r^+$ and $r^-$ with the corresponding ones from \eqref{def:generator} we get
\[
L(x,\dot{x})=\beta \left( \psi(\dot{x}) + \psi^*(\nabla E) + \dot{x} \nabla E\right), 
\]
and we formally prove Theorem~\ref{thm:LDF}.


\section{The Quadratic Limit}
\label{sec:quadlim}

In this section we precisely state and prove point A1 of Statement~\ref{statement:Connection} and the whole of Statement~\ref{statement:Commutation}. We are in the regime where $\beta \to 0$, $\alpha \to \infty$, with $\alpha \beta \to \omega$.

First we show heuristically why the functional $\J_{\beta}$ defined in \eqref{def:jbetanotalpha} is expected to converge to $\J_Q$ defined in \eqref{eq:defJ0}. Looking at the equation that minimises the functional $\J_{\beta}$, and doing a Taylor expansion for $\beta \ll 1$,
\[
\dot{x} = - 2 \alpha \sinh ( \beta \nabla E) \simeq -2 \alpha \beta \nabla E \to -2 \omega \nabla E.
\]
Considering the functional $\J_{\beta}$ for $\beta \ll 1$ it can be seen that 
\begin{equation*}
\psi^*_{\beta}(w)=\frac{2\alpha}{\beta} \left( \cosh(\beta w) - 1 \right) \simeq \alpha \beta w^2 \to \omega w^2,
\end{equation*}
\begin{equation*}
\psi_{\beta}(v)=\frac{v}{\beta}\log \left(\frac{v+\sqrt{v^2+4\alpha^2}}{2\alpha} \right)- \frac{1}{\beta} \sqrt{v^2+4\alpha^2} + \frac{2\alpha}{\beta} 
\simeq \frac{v^2}{4\alpha \beta},
\end{equation*}
implying that
\[ \psi^*_{\beta}(w)  \to \omega w^2, \; \qquad \psi_{\beta}(v) \to \frac{v^2}{4 \omega}. \]

We now turn to the rigorous proof of the convergence to the quadratic gradient flow and therefore point A1 of Statement~\ref{statement:Connection}. For this we need the concept of Mosco-convergence. Given a sequence of functionals $\phi_n$ and $\phi$ defined on a space $X$ with weak and strong topology, $\phi_n$ is said to \emph{Mosco-converge} to $\phi$ ($ \phi_{n} \stackrel{M}{\to} \phi$) in the weak-strong topology of $X$  if
\begin{equation}
\label{def:moscoconvergence}
\begin{cases}
\forall \; x_n \rightharpoonup x  \text{ weakly, } & \liminf \phi_n (x_n) \geq \phi(x), \\
\forall \;x \;\exists \; x_n \to x \text{ strongly such that } & \limsup \phi_n (x_n) \leq \phi(x).
\end{cases}
\end{equation}

The gradient-flow Definition~\ref{def:gengraflo} is based on the function space $AC(0,T)$. 
We define weak and strong topologies on the space $AC(0,T)$ by using the equivalence with $W^{1,1}(0,T)$.
Let $x,x_n \in AC(0,T)$. We say that $x_n$ converges weakly to $x$ ($x_n \rightharpoonup x$) if $x_n \to x$ strongly in $L^1(0,T)$ and $\dot{x}_n \rightharpoonup \dot{x}$ weakly in $L^1(0,T)$, i.e. in $\sigma(L^1,L^\infty)$; we say that $x_n$ converges strongly to $x$ ($x_n \to x$) if in addition $\dot{x}_n \to \dot{x}$ strongly in $L^1(0,T)$.

\begin{theorem}[Convergence to the quadratic limit]
\label{thm:quadlimit}
Given $E: \R \times [0,T] \to \R$ satisfying condition~\eqref{cond:energy}, for $x \in AC(0,T)$  consider the functional $\J_{\beta}$
\begin{equation*}
\J_{\beta}(x)= \int_0^T  \bigg( \psi_{\beta}(\dot{x}(t)) + \psi_{\beta}^*(-\nabla E(x(t),t)) + \dot{x}(t) \nabla E(x(t),t) \bigg)  dt,
\end{equation*}
then,  for $\alpha \to \infty$, $\beta \to 0$ and $\alpha \beta \to \omega >0$, $\J_{\beta} \xrightarrow{M} \J_Q$ in the weak-strong topology of $AC(0,T)$, with
\begin{equation}
\label{eq:defJ0}
\displaystyle \J_{Q}(x) :=  \int_0^T \left(  \frac{\dot{x}^2(t)}{4\omega}  + \omega (\nabla E)^2(x(t),t)  + \dot{x}(t) \nabla E(x(t),t) \right) dt.
\end{equation}
Moreover, if a sequence $\{ x_{\beta} \}$ is such that $\J_{\beta}(x_{\beta})$ is bounded, and
\begin{equation*}
E(x_{\beta}(0),0) + \int_0^T \partial_t E(x_{\beta}(s),s) \,ds \leq C \qquad \forall \beta,
\end{equation*}
then the sequence $\{ \dot{x}_{\beta} \}$ is compact in the topology $\sigma (L^1,L^{\infty})$.
\end{theorem}

\begin{proof}
First we prove the \emph{Mosco-convergence}.
The \emph{lim-sup} condition follows because, for $\beta \to 0$ 
\begin{equation*}
\psi_{\beta}(\eta) \to \frac{\eta^2}{4\omega}, \qquad  \psi^*_{\beta}(\xi) \to \omega \xi^2, \qquad \text{  locally uniformly.}
\end{equation*}
By the local uniform convergence we can choose the recovery sequence to be the trivial one.

Now we prove the \emph{lim-inf} inequality. The uniform convergence of $x_{\beta}$ to $x$ implies that we can pass to the limit in the terms $E(x_{\beta}(\cdot),\cdot)$ and $\int \partial_t E(x_{\beta},t) \,dt$. Then applying Fatou's lemma we find 
\begin{equation*}
\liminf_{\beta \to 0} \int_0^T \frac{2\alpha}{\beta}\left( \cosh(\beta \nabla E(x_{\beta},t))-1 \right) \,dt 
\geq \liminf_{\beta \to 0}  \int_0^T \alpha \beta (\nabla E)^2(x_{\beta},t) \,dt \geq  \int_0^T \omega (\nabla E)^2(x,t)\, dt,
\end{equation*}
where we used the inequality $2\cosh (\theta) \geq 2 + \theta^2$.
The function $\psi_{\beta}^*$ is non-increasing for $\beta \to 0$ so, for any $\beta \leq \overline{\beta}$, we have  $\psi_{\beta} \geq \psi_{\overline{\beta}}$. Then 
\begin{equation*}
\liminf_{\beta \to 0} \int_0^T \psi_{\beta}(\dot{x}_{\beta}) \,dt \geq \liminf_{\beta \to 0} \int_0^T \psi_{\overline{\beta}}(\dot{x}_{\beta}) \, dt \geq \int_0^T \psi_{\overline{\beta}}(\dot{x}) \, dt,
\end{equation*}
and we conclude taking the limit  $\overline{\beta}\to0$. 


Now we prove the \emph{compactness}. Let us suppose that  $\J_{\beta}(x_{\beta})$ and $E(x_{\beta}(0),0) + \int_0^T \partial_t E(x_{\beta},t)dt$  are bounded. Then, by the positivity of $\psi^*_{\beta}$, 
\begin{equation*}
\int_0^T \psi_{\beta}(\dot{x}_{\beta}) \, dt \leq C < \infty , \qquad \forall \beta.
\end{equation*}
With the choice $\overline{\beta}=1$ we have
\begin{equation*}
\int_0^T \psi_1(\dot{x}_{\beta}) \,dt \leq C \qquad \forall \beta \leq  1,
\end{equation*}
and the compactness in $\sigma(L^1,L^\infty)$ follows from the Dunford-Pettis theorem (e.g.~\cite[Th.~4.30]{Brezis11}).
\end{proof}

Note that this result can also be obtained by the abstract method of Mielke~\cite[Th.~3.3]{Mielke14TR}. Also note that the result can also be formulated in the weak-strict convergence of $BV$; for the lower semicontinuity this follows since the  weak convergence in $BV$ with bounded $\J_\beta$ implies weak convergence in $AC$, and for the recovery sequence it follows from our choice of the trivial sequence. 

\medskip

We end this section with two theorems completing the proof of Statement~\ref{statement:Commutation}, pictured in the right hand side of Figure~\ref{fig:ideaofpaper}. 

\medskip

First we define for each $h>0$ the SDE
\begin{equation}
\label{def:BMgraddrif} 
dY^h_t = -2 \omega \nabla E(Y^h_t,t) dt + \sqrt{2\omega h} dW_t,
\end{equation}
where $W_t$ is the Brownian motion on $\R$, $Y^h_0$ has law $\delta_{x_0}$ and its  generator $\Omega$ is defined as 
\begin{equation}
\label{eq:definitioOmegaDrift}
\Omega f (x) = -2 \omega \nabla E(x,t) \nabla f(x) + \omega h \Delta f(x).
\end{equation}
For $f \in C_b(\R)$ the expected value $\E$ is defined as
\[
\E(f(Y^h_t)|x_0)=\int_{\R}f(z)d\mu_t(z),
\]
with, denoting by $\Omega^T$ the adjoint of $\Omega$,
\[ 
\begin{cases}
\partial_t \mu_t=\Omega^T\mu_t= -2 \omega \nabla \cdot (\mu_t \nabla E  ) + \omega h \Delta \mu_t, \\
\mu_0=\delta_{x_0}.
\end{cases}
\]
where $\mu_t$ is the law at time $t$ of the process $Y^h$ started at the position $x_0$ at time $t=0$. Then the following theorems hold.

\begin{theorem}[Large deviations for the processes $X^n$ and $Y^h$]
\label{thm:lardevres}
Given an energy $E$ satisfying condition \eqref{cond:energy}, fix \/ $0 < \delta < 1$ and consider the sequence of processes $\{ X^n \}$ with generator\/ $\Omega_n$ defined in \eqref{def:generator} with $\beta=n^{-\delta}$ and\/ $\alpha \beta \to \omega$ for $n \to \infty$. Then, if $X^n(0) \to x_0$, the process $X^n$ satisfies a large-deviations principle in $D([0,T])$ with speed $n^{1-\delta}$ and with  rate function the extension of $\J_{Q}$ in~\eqref{eq:defJ0} to $BV$:
\begin{equation*}
\J_Q(x):=
\begin{cases} 
\displaystyle \int_0^T \left(  \frac{\dot{x}^2(t)}{4\omega}  + \omega (\nabla E)^2(x(t),t)  + \dot{x}(t) \nabla E(x(t),t) \right) dt & \text{ for } x \in AC(0,T), \\
+\infty & \text{otherwise}.
\end{cases}
\end{equation*}
Moreover, as $h\to0$ the process $Y^h$ defined in \eqref{def:BMgraddrif} satisfies a large-deviations principle in $D([0,T])$ with speed $h^{-1}$ and also with  rate function $\J_Q$.
\end{theorem} 

\begin{proof}

The proof of the large-deviation principle for $X^n$ relies on the fulfilment of three conditions, namely \emph{convergence of the operators} $H_n$, \emph{exponential tightness} for the sequence of processes $X^n$, and the \emph{comparison principle} for the limiting operator $H$, following the steps of \cite[Sec.~10.3]{FK06}.


We restrict ourselves, for sake of simplicity, to the case $\alpha = \omega / \beta$, and let $m = n^{1-\delta}$. To treat the time dependence we use the standard procedure of converting a time-dependent process into a time-independent process by adding the time to the state variable (see e.g. \cite[Sec.~4.7]{EthierKurtz86}): consider the variable $u=(x,t) \in \R \times [0,T]$, $f \in C^2_{c}(\R \times [0,T])$,  and given $\Omega_n$ defined in \eqref{def:generator}, we define  $Q_n$  as
\begin{equation}
\label{def:spacetimeprocess}
Q_n f(u)= \Omega_n f(x,t) + \partial_t f (x,t). 
\end{equation}
With $m=n^{1-\delta}$, we have,
\[
H_n f(u) = \frac{1}{m} e^{-mf(u)} Q_n e^{mf}(u) = \frac{1}{m} e^{-mf(u)} Q_n e^mf(u) + \partial_t f(u).
\]
Now, with the convention $u + 1 / n = (x + 1 / n , t)$,
\begin{align*}
H_n f(u) = & \;  \omega n^{2\delta}  \Biggl\{e^{-n^{-\delta} \nabla E(u)} \Bigl(e^{n^{1-\delta}(f(u+1/n)-f(u))} - 1\Bigr)
+ e^{n^{-\delta} \nabla E(u)} \Bigl(e^{n^{1-\delta}(f(u-1/n)-f(u))} - 1\Bigr)\Biggr\} + \partial_{t}f(u)\\
= & \; \omega n^{2\delta} \left[ \left(1-n^{-\delta}\nabla E(u) + o(n^{-\delta})\right) \left( n^{-\delta}\nabla f(u) + n^{-2\delta}\frac{1}{2}(\nabla f)^2(u) + o(n^{-2\delta}) \right) \right. \\
& \left.+ \left( 1+n^{-\delta}\nabla E(u) + o(n^{-\delta}) \right) \left(-n^{-\delta}\nabla f(u) + n^{-2\delta}\frac{1}{2}(\nabla f)^2(u) + o(n^{-2\delta}) \right) \right] + \partial_{t} f(u)\\
=& \; -2\omega\nabla E(u)\nabla f(u) + \omega (\nabla f)^2(u)  + \partial_{t} f (u) + o(1),
\end{align*}
implying  convergence in the uniform topology,
\[ 
\lim_{n \to \infty} \| H_n f - H f \|_{\infty}=0,
\]
to 
\[
H f (u) = - 2\omega \nabla E(u) \nabla f(u) + \omega (\nabla f)^2(u) + \partial_t f(u).
\]
The \emph{exponential tightness} holds by \cite[Cor.~4.17]{FK06}. 
The \emph{comparison principle} can be proved as in \cite[Example~6.11]{FK06}, modifying the definition of $\Phi_n$ with an additional time-dependent term
\[
\Phi_{n}(x,y,t,\tau)=\overline{f}(x,t)-\udf(y,\tau)-n\frac{(x-y)^2}{1+(x-y)^2}-n(t-\tau)^2,
\]
and then the proof, mutatis mutandis, follows similarly. 

Then the large-deviations principle holds in $D_{\R \times [0,T]}([0,T])$ with rate functional
\[
J_Q(u):=\int_0^T L(x,\dot{x},t,\dot{t}) \, ds,
\]
where $L$ is the Legendre transform of $H$ respect to the variables $(\nabla f,\partial_t f)$. It is just a calculation to check that
\[
L(x,\dot{x},t,\dot{t})=\frac{\dot{x}^2(s)}{2} + \nabla E (x(s),t(s)) \dot{x}(s) + \frac{\nabla E^2(x(s),t(s))}{2} + \mathbb{I}_{1}(\dot{t}(s)),
\]
where $\mathbb I_1$ is the indicator function of the set $\{1\}$, i.e.
\[
\mathbb I_{1}(\dot{t}(s))= \begin{cases}
0 & \dot{t}(s) = 1, \\
+\infty & \text{otherwise}.
\end{cases}
\] 
It is then clear that $J_Q(u) = \J_Q(x)$.

\medskip

The large-deviations result for $Y^h$ can be found in \cite[Th.~1.1 of Ch.~4]{FW12} in the case of a time-indedendent energy. The time-dependent case follows by the same modification as above.

\end{proof}

\begin{theorem}[Convergence to Brownian motion with gradient drift]
\label{thm:convprocesses}
Let be given an energy $E$ satisfying condition \eqref{cond:energy}, with $\nabla E$ uniformly continuous, let $\alpha_n \beta_n \to \omega$, $n \beta_n \to 1/h$, and let $\mu^n_t$ be the law of the process $\{ X^n(t) \}$ defined in \eqref{def:generator} with $\mu^n_0 = \delta_{X^n_0}$. If $X^n_0 \to x_0$, then $\mu_n$ weakly converge to $\mu$ (in the duality with $C_b(\R)$), where $\mu$ is the law of the Brownian motion with gradient drift \eqref{def:BMgraddrif} with $\mu_0=\delta_{x_0}$. 


\end{theorem}

\begin{proof}
This is a result of standard type, and we give a brief sketch of the proof for the case of time-independent $E$, using the semigroup convergence theorem of Trotter~\cite[Th.~5.2]{Tro58}. The assumptions of this theorem are satisfied by the existence of a single dense set on which $\Omega$ and $\Omega_n$ are defined, pointwise convergence of $\Omega_n$ to $\Omega$ on that dense set, and a dense range of $\lambda-\Omega$ for sufficiently large~$\lambda$. The assertion of Trotter's theorem is pointwise convergence of the corresponding semigroups at each fixed $t$, which implies convergence of the dual semigroups in the dual topology, which is the statement of Theorem~\ref{thm:convprocesses}.

We set the system up as follows. Define the state space $Y:=C_b(\overline \R)$ with the uniform norm, where $\overline\R$ is the one-point compactification of $\R$; define the core $D := \{f\in C^2_b(\R)\cap C(\overline \R): \Delta f \text{ uniformly continuous }\}$, which is dense in $Y$ for the uniform topology, and which will serve as the dense set of definition mentioned above for both $\Omega_n$ and $\Omega$. For each $f\in D$, $\Omega_nf\to\Omega f$ in the uniform topology. 

The density of the range of $\lambda - \Omega$ is the solvability in $D$ of the equation
\[
-\omega h\Delta f + 2\omega \nabla E\nabla f + \lambda f= g, \qquad \text{in }\R,
\]
for all $g$ in a dense subset of $Y$; we choose $g\in C_c(\R)+ \R$. This is a standard result from PDE theory, which can be proved for instance as follows. First note that we can assume $g\in C_c(\R)$, by adding a constant to both $g$ and $f$. Secondly, for sufficiently large $\lambda>0$ the left-hand side generates a coercive bilinear form in $H^1(\R)$ in the sense of the Lax-Milgram lemma, and therefore there exists a unique solution $f\in H^1(\R)$. By bootstrap arguments, using the continuity and boundedness of $\nabla E$, we find $f\in C^2_b(\R)$, and   since $f\in H^1(\R)\cap C^2_b(\R)$, $f(x)$ tends to zero at $\pm\infty$, implying that $f\in C_b^2(\R)\cap C(\overline\R)$. Finally, since $\nabla E$ is uniformly continuous, the same holds for $\Delta f$. This concludes the proof.
\end{proof}

%
%
%
%
%


\section{Rate-independent Limit}
\label{sec:rilimit}

In this section we prove point A2 of Statement~\ref{statement:Connection} and the whole of Statement~\ref{statement:Commutation}. We will prove point A2 with a theorem that holds in greater generality, without assuming the explicit form of the couple $\psi_{\beta}, \psi^*_{\beta}$, but only a few `reasonable' assumptions and the limiting behaviour. 

We are therefore in the regime where $\beta \to \infty$, $\log \alpha = -\beta A$ for some $A > 0$.

\subsection{Functions of bounded variation and rate-independent systems}
We now briefly recall the definition of the BV space of functions with bounded variation, following the notation of \cite{MRS12}. A full description of this space and its properties can be found in \cite{AFP00}. Given a function $x:[0,T] \to \R$ the total variation of $x$ in the interval $[0,T]$ is defined by
\begin{equation*}
\var (x,[0,T]):=\sup \left\{ \sum_{j=1}^n|x(t_j)-x(t_{j-1})|: 0=t_0 < t_1 < \cdots < t_{n-1} < t_n = T\right\}.
\end{equation*}
We say that $x \in BV([0,T])$ if $\var(x,[0,T]) < \infty$. The function $x$ then admits left and right limits $x(t_-)$ and $x(t_+)$ in every point $t \in [0,T]$, and we define the jump set of $x$ as
\[
J_x:=\bigg\{ t \in [0,T] : x(t_-)\neq x(t) \text{ or } x(t) \neq x(t_+) \bigg\},
\]
and the pointwise variation in the jump set as
\begin{equation}
\label{eq:normaljumpterm}
\pmuj(x,[0,T]):= \sum_{t \in J_x}  \Bigl(|x(t_-) - x(t)| + |x(t) - x(t_+)|\Bigr).
\end{equation}

The total variation admits the representation
\[
\var (x,[0,T])=\int_0^T |\dot{x}(t)|dt + \int_0^T d|Cx| + \pmuj(x,[0,T]),
\]
where $|\dot{x}|$ is the modulus of the absolutely continuous (a.c.) part of the distributional derivative of $x$; the measure $|Cx|$ is the Cantor part and $\pmuj$ represents the contribution of the (at most countable) jumps.

Given a sequence $\{ x_n \} \subset BV([0,T])$, we again define two notions of convergence. We say that $x_n$ \emph{weakly} converges to $x$ ($x_n \rightharpoonup x$) if $x_n(t)$ converges to $x(t)$ for every $t \in [0,T]$ and the variation is uniformly bounded, i.e. $\sup_n \var(x_n,[0,T]) < \infty$. We say that $x_n$ \emph{strictly} converges to $x$ ($x_n \to x$) if $x_n \rightharpoonup x$ and in addition $\var(x_n,[0,T])$ converges to $\var(x,[0,T])$ as $n \to \infty$.

According to the general setup of \cite{MielkeRossiSavare12a,MRS12} we define a notion of rate-independent system based on an energy balance similar to equation~\eqref{eq:gen_ene_equ}, where now the dissipation~$\psi$ has a linear growth, i.e. $\psi(\eta)=\psi_{RI}(\eta)=A|\eta|$ with $A > 0$.

We first define $\pmuj_E$, which can be viewed as an energy-weighted jump term, as 
\[
\pmuj_E(x,[0,T])=\sum_{t \in J_x}\Bigl[\Delta(x(t_-),x(t))+\Delta(x(t_+),x(t))\Bigr],
\]
where
\begin{equation}
\label{eq:jumpterm}
\Delta(x_0,x_1) = \inf \left\{ \int_0^1 \Big(|\nabla E(t,\theta(\tau))| \vee A \Big) |\dot{\theta}(\tau)|\,d\tau \; : \; \theta \in AC([0,1]), \; \theta(i)=x_i, \; i=0,1 \right\}.
\end{equation}
The relation between the definitions \eqref{eq:jumpterm} and \eqref{eq:normaljumpterm} becomes clear in the following inequality
\[
\pmuj_E(x,[0,T]) \geq A\;\pmuj(x,[0,T]),
\] 
where equality can be achieved depending on the behaviour of $E$, e.g. trivially when $|\nabla E(x,t)| \leq A$ for every $x,t$. Then we can interpret $\pmuj_E$ as a modified jump term, with an $E$-dependent weight. 

In analogy with the (generalized) gradient flow Definition~\ref{def:gengraflo} we define rate-independent systems. There is no unique way to define a rate-independent system. The so-called \emph{energetic solutions} have been introduced and analysed in \cite{MielkeTheil04,MielkeTheil99,MielkeTheilLevitas02}, and are based on the combination of a pointwise global minimality property and an energy balance. Here we concentrate on \emph{BV solutions}, as defined in \cite{MielkeRossiSavare12a}. Our limiting system will be of this type. 

Fix $A > 0$, the rate-independent dissipation $\psi_{RI}$ and its Legendre transform $\psi^*_{RI}$ are
\begin{equation}
\label{def:psiRIpsistar}
\psi_{RI}(v)=A|v|, \qquad \psi^*_{RI}(w)=\begin{cases}
0 & w \in [-A,A], \\
+ \infty & \text{otherwise}.
\end{cases}
\end{equation}

\begin{definition}[Rate-independent evolution, in the BV sense]
\label{def:RIgraflo}
Given an energy $E:\R \times [0,T] \to \R$, continuously differentiable, and $A>0$, a curve $x \in BV([0,T])$ is a \emph{rate-independent gradient flow} of $E$ in $[0,T]$ if it satisfies the  energy balance
\begin{multline}
\label{eq:RI_ene_equ}
\int_0^T \bigg( \psi_{RI}(\dot{x}(t)) + \psi^*_{RI}(-\nabla E (x,t)) \bigg) dt + A\int_0^T d|Cx| + \pmuj_E(x,[0,T]) \\
+ E(x(T),T) - E(x(0),0) - \int_0^T \partial_t E(x,t) dt = 0,
\end{multline}
with $\psi_{RI}$ and $\psi^*_{RI}$ are defined in \eqref{def:psiRIpsistar}.
\end{definition}

\subsection{Assumptions and the main result}

In the rest of this section we prove that the generalized gradient-flow evolution converges to the rate-independent one. This is point A2 of statement~\ref{statement:Connection}, formulated in Theorem~\ref{thm:RIgammaconv} showing Mosco-convergence of $\J_{\beta}$ to $\J_{RI}$, which is the left-hand side of~\eqref{eq:RI_ene_equ}. There are three main reasons why the convergence to a rate independent system should be expected. 

First, from a heuristic mathematical point of view, our choices of $\alpha$ and $\beta$ yield pointwise convergence of $\psi_{\beta}$ and $\psi_{\beta}^*$ to a one-homogeneous function and to its dual, the indicator function; this suggests a rate-independent limit. However, this argument does not explain which of the several rate-independent interpretations the limit should satisfy, nor does it explain the additional jump term. 

Secondly, from a physical point of view, the underlying stochastic model mimics a rate-independent system. This can be recognized by keeping the lattice size finite but letting  $\beta \to \infty$; then the rates either explode or converge to zero, depending on the value of $\nabla E$. We can interpret this in the sense that when a rate is infinite, with probability one a jump  will occur to the nearest lattice point with zero jump rate. 

Thirdly, considering the evolution, in the case $1 \ll \beta < \infty$ the generalized gradient flow will present fast transitions when $|\nabla E| > A$. By slowing down time during these fast transitions, we can capture what is happening at the small time scale of these fast transitions---which become jumps in the limit. This is exactly how we construct the recovery sequence in Theorem~\ref{thm:RIgammaconv}.

\medskip

The convergence will be proven in a greater generality; more precisely, we do not use the explicit formulas, but we require that $\psi_{\beta}$ and $\psi_{\beta}^*$  satisfy the following conditions: 
\begin{itemize}
\item[A] $\psi_{\beta}$ and $\psi_{\beta}^*$ are both symmetric, convex and $C^1$; \\
\item[B] $\psi_{\beta}^*$ converge pointwise to  $\psi_{RI}^*(w)= \begin{cases}		
+\infty & \text{ for } |w|>A, \\
0 & \text{ for } |w| \leq A; \end{cases}$ \\
\item[C] $\forall\, M>0 \ \exists \, \delta_{\beta} \to 0$ such that  as $\beta \to \infty$, \begin{align*}
& K_{\beta}^{-1}:=\partial \psi_{\beta}^* (A+\delta_{\beta}) \to \infty,\\
&\sup_{|w| \leq R} \; \frac{\partial \psi_{\beta}^* (w +M K_{\beta} )}{ \partial \psi_{\beta}^* ( w \vee(A+\delta_\beta))}K_{\beta} \to 0, \qquad \text{and}\\
& \partial \psi_{\beta}^* ( A + M K_{\beta} )K_{\beta} \to 0.
\end{align*}
\item[D]  For each $\alpha\geq1$ and for each $|w|\leq R$ there exists $\eta_\beta(w,\alpha)\geq0$ such that 
\[
\partial\psi_\beta^*(w+\eta_\beta(w,\alpha)) = \alpha \partial\psi_\beta^*(w),
\]
and $\eta_\beta$ is bounded uniformly in $\alpha$, $\beta$, and $|w|\leq R$.
\end{itemize}

It is important to underline that the previous conditions C-D are needed in Theorem~\ref{thm:RIgammaconv} only for the the $\Gamma$-limsup, meanwhile they are not necessary for the $\Gamma$-liminf. 

These conditions are satisfied by a large family of couples $\psi_{\beta}$-$\psi_{\beta}^*$. The two examples below show that our specific case and the vanishing-viscosity approach respectively are covered by the assumptions A-D.

\subsubsection*{Dissipation~\eqref{def:dissipation}: $\psi_{\beta}^*(w) = \beta^{-1}e^{-\beta A} \cosh (\beta w)$}

Conditions A and B are trivially satisfied. Then, considering only $w \geq A$ for simplicity, we get
\[
\partial \psi_{\beta}^*(w) = e^{-\beta A} \sinh (\beta w) \simeq e^{\beta(w-A)}.
\]
With the choice $\delta_{\beta} = \beta^{-1}\log(\beta)$ and $\lambda =1$, it holds that
\[ 
\partial \psi_{\beta}^* ( A +  MK_{\beta} )K_{\beta} \leq e^{\beta  MK_{\beta}}K_{\beta} \to 0,
\]
because  $K_{\beta} = \beta^{-1}$. Then condition C is satisfied with
\begin{equation*}
\frac{\partial \psi_{\beta}^* ( w +  K_{\beta} )}{ \partial \psi_{\beta}^* ( w )}K_{\beta} \leq \frac{\exp(\beta(w+  K_{\beta} - A))+1}{\exp(-\beta(A-w))}K_{\beta} \leq (\exp(\beta  K_{\beta})+1)K_{\beta} \to 0.
\end{equation*}
Condition D is satisfied because for $ w \gg 1$ we have that $\sinh(\beta w) \simeq \frac{1}{2} e^{\beta w}$. Then condition D approximately reads as
\[
e^{\beta(w+\eta)} \simeq \alpha e^{\beta w},
\]
which is satisfied for $\eta \simeq \beta^{-1} \log \alpha $.

\subsubsection*{Vanishing viscosity: $\psi_{\beta}^*(w) = \beta(|w|-A)_+^2$}

Also here, conditions A and B are immediately satisfied. Then, again considering $w \geq A$, we verify condition C by choosing $\delta_{\beta} \simeq \beta^{-1/3} $ and $\lambda = 1$, so that 
\[
\partial \psi_{\beta}^* (w) = 2\beta(w-A) \; \implies \; K_{\beta}^{-1}= 2\beta \delta_{\beta} \to \infty.
\]
Then it is just a calculation to check that condition C is satisfied in this case. Now condition D requires that 
\[
2\beta(w + \eta-A)=2\alpha \beta(w-A),
\]
and so $\eta=(\alpha-1)(w-A)$ satisfies the condition.

\begin{theorem}[Convergence to the rate-independent evolution]
\label{thm:RIgammaconv}
Given an energy $E$ satisfying condition \eqref{cond:energy}, a sequence of couple $\psi_{\beta}$-$\psi^*_{\beta}$ satisfying conditions A-D, and for $x \in BV([0,T])$ consider the functional $\J_\beta$
\begin{equation}
\label{def:Jbeta}
\J_{\beta}(x)= \begin{cases}
\displaystyle \int_0^T \bigg( \psi_{\beta}(\dot{x}(t)) + \psi_{\beta}^*(-\nabla E(x(t),t)) + \dot{x}(t) \nabla E(x(t),t) \bigg) \,dt & \text{for } x \in AC(0,T),\\
+\infty & \text{otherwise},
\end{cases}
\end{equation}
then, as $\beta \to \infty$, $\J_{\beta} \xrightarrow{M} \J_{RI}$ with respect to the weak-strict topology of $BV$, where $\J_{RI}$ is given by
\begin{multline}
\label{def:JRI}
\J_{RI}(x):=\int_0^T \bigg( \psi_{RI}(\dot{x}(t)) + \psi^*_{RI}(-\nabla E (x(t),t)) \bigg) dt + A\int_0^T d|Cx| + \pmuj_E(x,[0,T]) \\
+ E(x(T),T) - E(x(0),0) - \int_0^T \partial_t E(x(t),t) dt.
\end{multline} 
Moreover, if the sequence $\{ x_{\beta} \}$ is such that $\J_{\beta}(x_{\beta})$ is bounded, $\{ x_{\beta}(0)\}$ is bounded, and
\begin{equation}
\label{eq:condbound}
\int_0^t \partial_t E(x_{\beta}(s),s)ds \leq C \qquad \forall \beta, \; t \in [0,T],
\end{equation}
then $\{ x_{\beta} \}$ is weakly compact in $BV([0,T])$.
\end{theorem}

The proof is divided into three main steps. We first prove the compactness and the lim-inf inequality; this will follow as in \cite[Th.~4.1,4.2]{MRS12}. We report them for completeness and we translate their proof because we can avoid some technicalities. To finish the proof we need to construct a recovery sequence. When minimizers with $\J_{RI}\equiv 0$ are considered, then the recovery sequence is easy to construct; we just need to take a sequence $x_{\beta}$ such that  $\J_{\beta}(x_{\beta}) = 0$ for every~$\beta$. But for the full Mosco-convergence, we need to find a way to construct a recovery sequence also for non-minimizers of $\J_{RI}$. This is the last part of the proof and it will be achieved using a parametrized-solution technique. 

\subsection{Proof of compactness and the lower bound}
\begin{proof}[Proof of compactness]
Recall that weak convergence in $BV$ is equivalent to pointwise convergence supplemented with a global bound on the total variation (e.g.~\cite[Prop.~3.13]{AFP00}).

First we show that $|x_{\beta}(t)-x_{\beta}(0)|$ is bounded uniformly in $t$ and $\beta$. We  observe that  
\begin{equation}
\label{ineq:psi-from-below}
\psi_{\beta}(v) + \psi^*_{\beta}(A) \geq Av ,
\end{equation}
and so we obtain
\begin{equation*}
A|x_{\beta}(t)-x_{\beta}(0)| \leq A \int_0^t |\dot{x}_{\beta}|\,ds \leq \int_0^t \psi_{\beta}(\dot{x}_{\beta})\,ds + t \sup_{\beta} \psi^*_{\beta}(A) \leq C  < + \infty,
\end{equation*}
where the constant $C$ may change from line to line. Then
\begin{equation*}
|x_{\beta}(t)-x_{\beta}(0)| \leq C \; \text{ for every $\beta$ and every }t \in [0,T]. 
\end{equation*}
The inequality above and the boundedness of $x_{\beta}(0)$ imply that the whole sequence is bounded for every $t \in [0,T]$. 

Next we show the existence of a converging subsequence. 

For every $0 \leq t_0 \leq t_1 \leq T$ we recall the bound
\begin{equation*}
A|x_{\beta}(t_1)-x_{\beta}(t_0)| \leq \int_{t_0}^{t_1}A|\dot{x}_{\beta}| \, dt \leq \int_{t_0}^{t_1} \left( \psi_{\beta}(\dot{x}_{\beta})+\psi_{\beta}^*(A) \right)dt.
\end{equation*}
Defining the non-negative finite measures on $[0,T]$
\begin{equation*}
\nu_{\beta,A}:= \left( \psi_{\beta}(\dot{x}_{\beta})+\psi_{\beta}^*(A) \right)\L^1, 
\end{equation*}
up to extracting a suitable subsequence, we can suppose that they weakly-$*$ converge to a finite measure $\nu_A$, so that
\begin{equation*}
A\limsup_{\beta \to \infty}|x_{\beta}(t_0)-x_{\beta}(t_1)| \leq \limsup_{\beta \to \infty} \nu_{\beta,A}([t_0,t_1])\leq \nu_A([t_0,t_1]).
\end{equation*}
Defining the jump set $J:= \{ t \in [0,T] : \nu_a( \{ t \} ) > 0 \}$ and considering a countable set $I \supset J$ that is dense in $[0,T]$, we can find a subsequence $\beta_h$ such that  $ x_{\beta_h} \stackrel{p.w.}{\longrightarrow} x$ for every $t \in I$ as $\beta_h \to \infty$. From now on, for simplicity, we will number the subsequence with the same index of the main sequence. Then
\begin{equation}
\label{eq:inequalityboundBV}
A|x(t_0)-x(t_1)| \leq \nu_A([t_0,t_1]), \qquad \text{ for every } t_0,t_1 \in I.
\end{equation}
The curve $I \ni t \mapsto x(t)$ can be uniquely extended to a continuous curve in $[0,T]\setminus J$, that we will still denote by $x$. Arguing by contradiction we show that the whole $x_{\beta}(t)$ converges pointwise to $x(t)$. If the pointwise convergence does not hold, then there will be a further subsequence $t_{\beta_n} \to t \in [0,T]\setminus J$ such that $x_{\beta_n}(t_{\beta_n}) \to \tilde{x} \neq x(t)$, but this is in contradiction to the previous inequality
\begin{equation*}
A|x(t)-\tilde{x}| \leq \liminf_{\beta_n \to \infty}A|x_{\beta_n}(t)-x_{\beta_n}(t_{\beta_n})| \leq \limsup_{\beta_n \to \infty} \nu_{\beta_n,A}([t,t_{\beta_n}])=\nu_A(\{ t \}) = 0.
\end{equation*}
We have so proven the pointwise convergence of $x_{\beta}$ to $x$; the inequality~\eqref{eq:inequalityboundBV} then gives a uniform bound on the BV norm of $x_{\beta}$ and so we conclude.
\end{proof}

\begin{proof}[Proof of the lim-inf inequality]
Let $\{ x_{\beta} \} \subset AC(0,T)$ be a sequence such that $ \J_{\beta}(x_{\beta})$ is bounded, and which converges weakly to $x \in BV([0,T])$. By the arguments above $x_{\beta}$ is bounded uniformly in $t$ and $\beta$ and every term in the functional is bounded itself by a constant independent of $\beta$. The following limits, follow from the pointwise convergence and  Lebesgue's dominated convergence theorem:
\begin{equation*}
 E(x_{\beta}(\cdot),\cdot) \to E(x(\cdot),\cdot), \qquad \int_0^T \partial_t E(x_{\beta}(t),t)\,dt \to \int_0^T \partial_t E(x(t),t)\,dt.
\end{equation*}
As we said, the integral $\int_0^T \psi_{\beta}^*(-\nabla E(x_{\beta}(t),t))dt$  is bounded for every $\beta$ by a constant that we  still denote by $C$. Because of the monotonicity of $\psi_{\beta}^*$ this bound implies that
\begin{equation*}
\psi_{\beta}^*(a)\L^1\{t\in(0,T): |\nabla E(x_{\beta}(t),t)| \geq a \} \leq K \qquad \forall a \geq 0,
\end{equation*}
and since $\psi^*_{\beta}(w)\to +\infty \text{ for } |w| > A$,  we obtain 
\begin{equation*}
\lim_{\beta \to \infty} \L^1\{t\in(0,T): |\nabla E(x_{\beta}(t),t)| \geq a \} = 0 \qquad \forall a > A.
\end{equation*}
This proves that $|\nabla E(x(t),t)| \leq A$ a.e.\ and therefore $\int_0^T \psi_{RI}^*(\nabla E(x(t),t))\, dt=0$; it trivially follows that
\begin{equation*}
\liminf_{\beta \to \infty} \int_0^T \psi_{\beta}^*(\nabla E(x_{\beta}(t),t)) \, dt \geq \int_0^T \psi_{RI}^*(\nabla E(x(t),t)) \, dt.
\end{equation*}

We now prove the second part of the inequality,
\begin{equation*}
\liminf_{\beta \to \infty} \int_0^T \left(  \psi_{\beta}(\dot{x}) + \psi^*_{\beta}(\nabla E)\right)dt \geq \int_0^T A|\dot{x}| \, dt + A\int_0^T d|Cx| + \pmuj_E (x,[0,T]).
\end{equation*}
As in the proof of the compactness, we consider the non-negative finite measure on $[0,T]$
\begin{equation*}
\nu_{\beta} := \left( \psi_{\beta}(\dot{x}_{\beta}) + \psi_{\beta}^*(\nabla E(x_{\beta},\cdot)) \right)\L^1 ,
\end{equation*}
up to extracting a subsequence, we can suppose that they weakly$^*$ converge to a finite measure
\[
\nu_0 + \psi_{RI}^*(\nabla E( x,\cdot))\L^1.
\]
Because $|\nabla E(x,\cdot)| \leq A$ $\L$-a.e.\ we obtain that, as in the proof of the compactness,
\begin{equation*}
\nu_0 \geq A(|\dot{x}|+|Cx|+|Jx|) =  \psi_{RI}(\dot{x}) + A(|Cx|+|Jx|) .
\end{equation*}

This inequality is slightly too weak for us. The Cantor part and the Lebesgue measurable part are fine, but we need a stronger characterization of the jump part: for all $t \in J_x$,
\begin{equation*}
\nu_0(\{t\}) \geq \Delta(x(t_-),x(t))+\Delta(x(t),x(t_+)).
\end{equation*}
To prove this, fix $t\in J_x$ and take two sequences $h_{\beta}^- < t < h_{\beta}^+$ converging monotonically to $t$ such that
\begin{equation*}
x_{\beta}(h_{\beta}^-) \to x(t_-), \qquad x_{\beta}(h_{\beta}^+) \to x(t_+),
\end{equation*}
and  define 
\begin{equation*}
s_{\beta}(h):=h+\int_t^h \left( \psi_{\beta}(\dot{x}_{\beta}(\tau))+\psi_{\beta}^*(\nabla E( x_{\beta}(\tau),\tau))  \right)d\tau; \qquad s_{\beta}^{\pm}:=s_{\beta}(h^{\pm}_{\beta}).
\end{equation*}
Because of the convergence of $\nu_{\beta}$ we have
\begin{equation*}
\limsup_{\beta \to \infty}(s_{\beta}^+ - s_{\beta}^-) \leq \limsup_{\beta \to \infty} \nu_{\beta}([s_{\beta}^-,s_{\beta}^+]) \leq \nu_0(\{ t \}),
\end{equation*}
and up to extracting a subsequence we can assume that $s_{\beta}^{\pm} \to s^{\pm}$. Denote by $h_{\beta}:=s_{\beta}^{-1}$ the inverse map of $s_{\beta}$, we observe that $h_{\beta}$ is 1-Lipschitz and  monotone, and it maps $[s_{\beta}^-,s_{\beta}^+]$ onto $[h_{\beta}^-,h_{\beta}^+]$. We can then define the following Lipschitz functions
\begin{equation*}
\theta_{\beta}(s):= \begin{cases}
x_{\beta}(h_{\beta}(s)) & \text{if } s \in [s_{\beta}^-,s_{\beta}^+], \\
x_{\beta}(h_{\beta}^+) & \text{if }  s \geq s_{\beta}^+, \\
x_{\beta}(h_{\beta}^-) & \text{if }  s \leq s_{\beta}^-. \\
\end{cases}
\end{equation*}
The functions $\theta_{\beta}$ are uniformly Lipschitz, since (writing $\tau= h_\beta(s)$)
\[
|\dot \theta_\beta(s)| = |\dot x_\beta(h_\beta(s)) |\,|\dot h_\beta(s)|
\leq \frac{|\dot x_\beta(\tau) |}{1+ \psi_{\beta}(\dot{x}_{\beta}(\tau))+\psi_{\beta}^*(\nabla E( x_{\beta}(\tau),\tau))} 
\stackrel{\eqref{ineq:psi-from-below}}\leq \frac{\max\{c,1\}}A,
\]
and they take the special values 
\[
\theta_{\beta}(s_{\beta}^{\pm})=x_{\beta}(h_{\beta}^{\pm}), \qquad \theta_{\beta}(t)=x_{\beta}(t).
\]
Therefore, denoting by $I$ a compact interval containing the intervals $[s_{\beta}^-,s_{\beta}^+]$ for all $\beta$, then up to a subsequence, we have that 
\[ 
\theta_{\beta}(s) \to \theta(s), \qquad |\dot{\theta}_{\beta}| \stackrel{*}{\rightharpoonup }m \text{ in }L^{\infty}(I) \quad \text{with } m\geq |\dot{\theta}|.
\]
Moreover $\theta(s^{\pm})=x(t_{\pm})$ and $\theta(t)=x(t)$. Then using the inequality 
\begin{equation*}
\psi_{\beta}(v)+\psi_{\beta}^*(w) \geq (|w| \vee A)v - \psi_{\beta}^*(A),
\end{equation*}
we obtain
\begin{equation*}
\begin{split}
\nu_0(\{t\}) & \geq \limsup_{\beta \to \infty} \int_{h_{\beta}^-}^{h_{\beta}^+} \Big(  \psi_{\beta}(\dot{x}_{\beta}(\tau))+\psi_{\beta}^*(\nabla E(x_{\beta}(\tau),\tau))\Big)\,d\tau \\
& \geq \liminf_{\beta \to \infty} \int_{h_{\beta}^-}^{h_{\beta}^+} \Big(  (|\nabla E(x_{\beta}(\tau),\tau)| \vee A)|\dot{x}_{\beta}|(\tau) - \psi_{\beta}^*(A) \Big) \,d\tau \\
& \geq \liminf_{\beta \to \infty} \int_{s_{\beta}^-}^{s_{\beta}^+}  (|\nabla E(\theta_{\beta}(s),h_{\beta}(s))| \vee A)|\dot{\theta}_{\beta}|(s) \, ds - (h_{\beta}^+ -h_{\beta}^-)\psi_{\beta}^*(A).
\end{split}
\end{equation*}
The last term $(h_{\beta}^+ -h_{\beta}^-)\psi_{\beta}^*(A)$ tends to zero as $\beta\to\infty$. Therefore
\begin{equation*}
\begin{split}
\nu_0( \{ t \} ) & \geq \liminf_{\beta \to \infty}\int_I  (|\nabla E(\theta_{\beta}(s),h_{\beta}(s))| \vee A)|\theta_{\beta}|(s)\,  ds \\
& \stackrel{(*)}\geq \int_I   (|\nabla E(\theta(s),s)| \vee A)m(s) \, ds \geq  \int_{s^-}^{s^+} (|\nabla E(\theta(s),s)| \vee A)|\dot{\theta}|(s) \, ds \\
& \geq  \Delta(x(t_-),x(t))+\Delta(x(t),x(t_+)).
\end{split}
\end{equation*}
The inequality~$(*)$ follows from the technical Lemma~\cite[Lemma~4.3]{MRS12}, and so we conclude. \end{proof}

\subsection{Proof of the lim-sup inequality}

\begin{proof}
We assume that we are given ${x} \in BV([0,T])$; we will construct a sequence $x_\beta$ such that $\J_{\beta} (x_\beta) \to \J_{RI}( x)$.

\medskip
\textit{Reparametrization.}
A central tool in this construction is a reparametrization of the curve $ x$ (as in Figure~\ref{fig:timeparametrization}), in terms of a new time-like parameter $s$ on a domain $[0,S]$. The aim is to expand the jumps in $x$ into smooth connections. 

As in~\cite[Prop.~6.10]{MielkeRossiSavare12a}, we define
\[
\ss(t):= t + \int_0^t \left( \psi_{RI}(\dot{{x}}) + \psi_{RI}^*(\nabla E({x},\tau)) \right)d\tau + A\int_0^t d|C{x}| + \pmuj_E({x},[0,t]),
\] 
then there exists a Lipschitz parametrization $(\st,\sx):[0,S]\to[0,T]\times \R \;$ such that $\st$ is non-decreasing,
\begin{equation}
\label{eq:inverserelation}
\st(\ss(t))= t, \qquad \text{and} \qquad 
\sx(\ss(t))={x}(t) \text{ for every } t \in [0,T],
\end{equation}
and such that
\begin{equation}
\label{eq:equivalenceparametrization}
\int_0^S \mathbb{L} (\sx,\st,\dot{\sx},\dot{\st})\,ds = \int_0^T \left( \psi_{RI}(\dot{{x}}) + \psi_{RI}^*(\nabla E({x},\tau)) \right)d\tau + A\int_0^T d|C{x}| + \pmuj_E({x},[0,T]),
\end{equation} 
where 
\[ 
\mathbb{L}(\sx,\st,\dot{\sx},\dot{\st}) =\begin{cases}
A|\dot{\sx}|+\psi_{RI}^*(|\nabla E(\sx,\st)|) & \text{if } \dot{\st} > 0, \\
|\dot{\sx}|\bigl(A \vee |\nabla E(\sx,\st)|\bigr) & \text{if } \dot{\st} = 0.
\end{cases}
\]
Moreover, it also holds that
\begin{equation}
\label{eq:equivalenceVar}
\var(\sx,[0,S]) = \var(x,[0,T]).
\end{equation}
Note that $\mathbb L(\sx,\st,\dot{\sx},\dot{\st}) \geq |\dot{\sx}|\bigl(A \vee |\nabla E(\sx,\st)|\bigr)$, since $\psi_{RI}^*(w)$ is only finite when $|w|\leq A$.

\medskip
\textit{Preliminary remarks.}
The third term in $\J_\beta(x_\beta)$ (see~\eqref{def:Jbeta}) is equal to
\[
E( x_\beta (T),T) - E(x_\beta(0), 0) - \int_0^T \partial_t E( x_\beta(t),t)\,dt,
\]
and these three terms pass to the limit under the strict convergence $x_\beta\to x$ that we prove below. We therefore focus on the other terms in $\J_\beta$ and $\J_{RI}$. By~\eqref{eq:equivalenceparametrization} it is sufficient to prove that
\begin{equation}
\label{limsup-thm:ineq-to-prove}
\limsup_{\beta\to\infty} \int_0^T \left[ \psi_{\beta}\left( \dot{x}_{\beta}(t) \right) + \psi_{\beta}^*\left( \nabla E(x_{\beta}(t),t) \right) \right]\,dt
\leq \int_0^S \mathbb{L} (\sx(s),\st(s),\dot{\sx}(s),\dot{\st}(s))\,ds.
\end{equation}

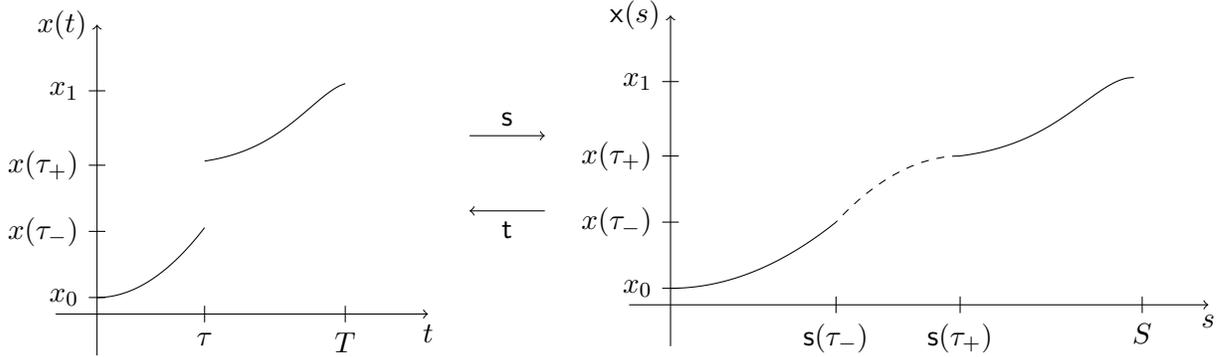
\begin{figure}[b]
\centering
\begin{tabular}{c c c}
\begin{tikzpicture}[scale=1.1]
 \draw[->] (-.5,0) -- (4,0);
 \draw[->] (0,-.5) -- (0,3.5);
 \draw[domain=0:1.3,smooth,variable=\x] plot ({\x},{0.5*\x^2+0.2});
 \draw (-.1,.2) -- (.1,.2);
 \draw (-.1,.2) node[anchor=east] {$x_0$};
 \draw (-.1,1) -- (.1,1);
 \draw (-.1,1) node[anchor=east] {$ x(\tau_-)$};
 \draw (-.1,1.8) -- (.1,1.8);
 \draw (-.1,1.8) node[anchor=east] {$ x(\tau_+)$};
 \draw (-.1,2.7) -- (.1,2.7);
 \draw (-.1,2.7) node[anchor=east] {$ x_1$};
 \draw[domain=1.3:3,smooth,variable=\x] plot ({\x},{5(sin(\x^4)+1.8});
 \draw (1.3,-.1) -- (1.3,.1); 
 \draw (1.3,-.1) node[anchor=north] {$\tau$};
 \draw (0,3.5) node[anchor=east] {$ x(t)$};
 \draw (4,0) node[anchor=north] {$t$};
 \draw (3,.1) -- (3,-.1);
 \draw (3,-.1) node[anchor=north] {$T$};
\end{tikzpicture}
&
\begin{tikzpicture}[scale=0.5]
 \draw (0,-2) {};
 \draw[->] (0,4) -- (2,4);
 \draw (1,4) node[anchor=south] {$\ss$};
  \draw[->] (2,2) -- (0,2);
 \draw (1,1) node[anchor=south] {$\st$};
\end{tikzpicture}
&
\begin{tikzpicture}[scale=1.1]
 \draw[->] (-.5,0) -- (6.5,0);
 \draw[->] (0,-.5) -- (0,3.5);
 \draw[domain=0:2,smooth,variable=\x] plot ({\x},{0.2*\x^2+0.2});
 \draw (-.1,.2) -- (.1,.2);
 \draw (-.1,.2) node[anchor=east] {$x_0$};
 \draw (-.1,1) -- (.1,1);
 \draw (-.1,1) node[anchor=east] {$ x(\tau_-)$};
 \draw (-.1,1.8) -- (.1,1.8);
 \draw (-.1,1.8) node[anchor=east] {$ x(\tau_+)$};
 \draw (-.1,2.7) -- (.1,2.7);
 \draw (-.1,2.7) node[anchor=east] {$x_1$};
 \draw[domain=1.3:3.4,smooth,variable=\x] plot ({\x+2.2},{5(sin(((0.85*\x  + 0.2) )^4) + 1.75});
 \draw[dashed]      (3.5,1.8) parabola  (2,1);
 \draw (2,-.1) -- (2,.1); 
 \draw (2,-.1) node[anchor=north] {$\ss(\tau_-)$};
 \draw (3.5,-.1) -- (3.5,.1);
 \draw (3.5,-.1) node[anchor=north] {$\ss(\tau_+)$};
 \draw (0,3.5) node[anchor=east] {$\sx(s)$};
 \draw (6.5,0) node[anchor=north] {$s$};
 \draw (5.7,.1) -- (5.7,-.1);
 \draw (5.7,-.1) node[anchor=north] {$S$};
\end{tikzpicture}
\\
\end{tabular}
\caption{Schematic representation of the time parametrization procedure. The curve $\sx$ is such that $\sx(\ss(t))=x (t)$.}
\label{fig:timeparametrization}
\end{figure}


From condition \eqref{cond:energy} we have that $\nabla E(x,t)$ is uniformly Lipschitz continuous in $t$; let $L$ be the Lipschitz constant. In order to define a time rescaling we introduce an auxiliary function.
We fix 
\[ 
M := L \int_0^S |\dot \sx(s)|\, ds = L \var(x,[0,T]) = L \var(\sx,[0,S]),
\]
and use Hypothesis C to obtain sequences $\delta_\beta,K_\beta\to0$ for this value of $M$. We now define
\begin{align}
\label{eq:bipotential}
p_{\beta}(v,w)&:=\inf_{\e > 0}\left\{ \e \psi_{\beta}\Bigl(\frac{v}{\e}\Bigr)+\e \psi_{\beta}^*\bigl(|w| \vee (A+\delta_{\beta})\bigr) \right\}\\
&= |v| \left( |w| \vee (A+\delta_{\beta}) \right), \notag
\end{align}
and the infimum is achieved by 
\begin{equation}
\label{def:epsilon}
{\e_{\beta}(v,w)}:= \frac{v}{\partial \psi_{\beta}^* \bigl( |w| \vee (A+\delta_{\beta}) \bigr)} .
\end{equation}
The function $\e_\beta$ can be interpreted as an optimal time rescaling of a given speed $v$ and a given force $|w|\vee (A+\delta_\beta)$.

\medskip
\textit{Definition of the new time $\st_\beta$ and the recovery sequence $x_\beta$.}
For sake of simplicity, in the following we construct a recovery sequence only for a curve $x$ with jumps at $0$ and $T$. Later in the proof we show that, in a similar way, a recovery sequence can be constructed for a curve $x$ with countable jumps with transparent changes in the proof. 

We construct the recovery sequence by first perturbing the time variable $\st$.
We define $\st_{\beta}:[0,S]\to[0,T_{\beta}]$ as the solution of the differential equation
\begin{equation*}
\dot{\st}_{\beta}(s)= \dot{\st}(s) \vee \e_{\beta}\Bigl(\dot{\sx}(s),\nabla E\bigl(\sx(s),\st(s)\bigr)\Bigr),
\qquad \st_\beta(0) = 0.
\end{equation*}
We can assume that $|\dot \sx (s)| \neq 0$ for $s \in [0,\ss(0)]$ and $s \in [\ss(T^-),S]$ to guarantee the positivity of $\e_\beta$. Then, for $s \in [\ss(0),\ss(T^-)]$, we have  
\[
\dot \st (s) = \left. \frac{1}{\dot \ss(t)} \right|_{t=\st(s)} > 0 \, ,
\]
so that $\dot\st_\beta(s)>0$ for all $s \in [0,S]$. The range of $\st_\beta $ is $[0,T_\beta]$, with $T_\beta\geq T$; since the recovery sequence $x_\beta$ is to be defined on the interval $[0,T]$, we rescale $\st_\beta$ by
\[
\lambda_\beta := \frac{T_\beta}T \geq 1,
\]
and define our recovery sequence as follows:
\begin{equation*}
x_{\beta}(t):=\sx \left( \st_{\beta}^{-1} \left( t \lambda_{\beta} \right) \right), \qquad \qquad \text{so that}\qquad 
\dot{x}_{\beta}(t)=\frac{\dot{\sx}}{\dot{\st}_{\beta}}\left(  \st_{\beta}^{-1} \left( t \lambda_{\beta} \right) \right) \lambda_{\beta}. 
\end{equation*}
We now have that
\begin{align*}
\int_0^T \Bigl[ \psi_{\beta}\left( \dot{x}_{\beta}(t) \right) &+ \psi_{\beta}^*\left( \nabla E(x_{\beta}(t),t) \right) \Bigr]\,dt \\
&= \int_0^T \left[ \psi_{\beta}\left( \frac{\dot{\sx}}{\dot{\st}_{\beta}} \left( \st_{\beta}^{-1} \left( t\lambda_{\beta} \right)\right) \lambda_{\beta}  \right) + \psi_{\beta}^*\left( \nabla E\bigl(\sx ( \st_{\beta}^{-1} ( t\lambda_{\beta} )),t\bigr) \right) \right] dt \\
&=\int_0^S\left[ \psi_{\beta} \left( \frac{\dot{\sx}}{\dot{\st}_{\beta}}(s)\lambda_{\beta} \right) + \psi_{\beta}^* \Bigl( \nabla E(\sx(s),\st_{\beta}(s)\lambda_{\beta}^{-1}) \Bigr) \right]\frac{\dot{\st}_{\beta}(s)}{\lambda_{\beta}}\,ds.
\end{align*}

\medskip
\textit{Estimates.}
The inequality~\eqref{limsup-thm:ineq-to-prove} now follows from the following three estimates:
\begin{lemma}
\label{lemma:thm-limsup-estimates}
Write $\e_\beta(s) := \e_{\beta}\Bigl(\dot{\sx}(s),\nabla E\bigl(\sx(s),\st(s)\bigr)\Bigr)$. Then there exists $C_\beta \to 0$ for $\beta \to \infty$ such that
\begin{align}
\label{l-limsup-est:1}
& \int_0^S \biggl[\psi_\beta^*\Bigl(|\nabla E(\sx(s),\st_\beta(s)\lambda_\beta^{-1})|\Bigr) \frac{\dot \st_\beta(s)}{\lambda_\beta} -  \psi_\beta^*\Bigl(|\nabla E(\sx(s),\st(s))|\Bigr)\e_\beta(s)\biggr] \,ds \leq  C_\beta S ;\\
\label{l-limsup-est:2}
& \int_0^S \biggl[ \psi_{\beta} \biggl( \frac{\dot{\sx}}{\dot{\st}_{\beta}}(s) \lambda_{\beta} \biggr) \frac{\dot \st_\beta(s)}{\lambda_\beta} 
- \psi_{\beta} \biggl( \frac{\dot{\sx}}{\e_\beta}(s) \biggr)\e_\beta(s)\biggr] \, ds \leq C_\beta S; \\
& \begin{aligned}
\label{l-limsup-est:3}
\int_0^S  \biggl[ \psi_{\beta} \biggl( \frac{\dot{\sx}}{\e_\beta}(s) \biggr) + \psi_\beta^*\Bigl(|\nabla &E(\sx(s),\st(s))|\Bigr)\biggr] \e_\beta(s) \, ds  \\
&\leq \int_0^S |\dot\sx(s)| \bigl(A\vee|\nabla E(\sx(s),\st(s))|\bigr)\, ds + C_\beta S.
\end{aligned}
\end{align}

\end{lemma}
We prove this Lemma below.

\medskip
\textit{Convergence and conclusion.} Strict convergence of $x_{\beta} \to x$ follows if we prove the pointwise convergence $x_{\beta}(t) \to x(t)$ for all $t \in [0,T]$ and the convergence of the variation. \\
Recall the definition of $\dot{\st}_{\beta}(s)= \dot{\st}(s) \vee \e_{\beta}\Bigl(\dot{\sx}(s),\nabla E\bigl(\sx(s),\st(s)\bigr)\Bigr)$, considering that 
\[
\lim_{\beta \to \infty}\sup_{s}\e_{\beta}(s) = \lim_{\beta \to \infty} \sup_s \frac{  |\dot{\sx}(s)|}{\partial \psi^*_\beta (|\nabla E(\sx(s),\st)| \vee (A + \delta_\beta)} \leq \lim_{\beta \to \infty} \frac{ \sup_{s} |\dot{\sx}(s)|}{\partial \psi^*_\beta (A + \delta_\beta)} = 0,
\]
it implies $\dot{\st}_{\beta}(s) \to \dot{\st}(s)$, and so it also holds
\[
\st_{\beta}(s) \to \st(s) \qquad \implies \qquad
\st_{\beta}^{-1}(t \lambda_{\beta}) \to \ss(t) \qquad \forall t \in (0,T).
\]
Moreover, $\dot \st_\beta(s) > 0$ implies that $\st_\beta^{-1} (0) = 0$ and $\st_\beta^{-1} (T_\beta) = S$, and so we have that
\[
x_{\beta}(t)= \sx \left( \st_{\beta}^{-1} \left( t \lambda_{\beta} \right)  \right) \to \sx(\ss(t))\stackrel{\eqref{eq:inverserelation}}= x(t) \qquad \forall \, t \in [0,T].
\]
The convergence of the variation is automatic, since by definition of $x_{\beta}$
\begin{equation*}
\int_0^T |\dot{x}_{\beta}(t)|\,dt = \int_0^S |\dot{\sx}(s)| ds = \var(\sx,[0,S]) \stackrel{\eqref{eq:equivalenceVar}}{=} \var(x,[0,T]).
\end{equation*} 

\textit{Recovery sequence for a general curve $x$.}
Now we show how to construct a recovery sequence for a curve with countable jumps. Given the jump set $J_x$, consider a countable set $\{ t^i \} \supset J_x$ (with $t^i < t^{i+1}$) such that the interval $[0,T]$ can be written as the union of disjoint subintervals
\[
[0,T]=\bigcup_i \Sigma^i  \qquad \text{ where } \Sigma^i = [t^i,t^{i+1}].
\]
Then, let $t_\beta^i=\st_\beta(\ss(t^i))$, we define
\[
\lambda_{\beta}^i=\frac{t^{i+1}_\beta - t^i_\beta}{t^{i+1}-t^i},
\]
and the recovery sequence is
\begin{equation}
\label{eq:def_gen_rec_seq}
x_{\beta}(t):=\sx \left( \st_{\beta}^{-1} \left( \lambda^i_\beta(t - t^i) + t^i_\beta \right) \right) \qquad \text{ for } t \in \Sigma^i,
\end{equation} 
so that
\begin{equation*}
\dot{x}_{\beta}(t)=\frac{\dot{\sx}}{\dot{\st}_{\beta}}\left(  \st_{\beta}^{-1} \left(  \lambda^i_\beta(t - t^i) + t^i_\beta \right) \right) \lambda^i_{\beta} \qquad \text{ for } t \in (\Sigma^i)^\circ.
\end{equation*}
We have now that
\begin{align*}
\int_0^T \Bigl[ &\psi_{\beta}\left( \dot{x}_{\beta}(t) \right) + \psi_{\beta}^*\left( \nabla E(x_{\beta}(t),t) \right) \Bigr]\,dt \\
&= \sum_i \int_{\Sigma^i} \left[ \psi_{\beta}\left( \frac{\dot{\sx}}{\dot{\st}_{\beta}} \left( \st_{\beta}^{-1} \left( \lambda^i_\beta(t - t^i) + t^i_\beta \right)\right) \lambda^i_{\beta}  \right) + \psi_{\beta}^*\left( \nabla E\bigl(\sx ( \st_{\beta}^{-1} (  \lambda^i_\beta(t - t^i) + t^i_\beta )),t\bigr) \right) \right] dt \\
&=\sum_i \int_{\ss(t^i)}^{\ss(t^{i+1})}\left[ \psi_{\beta} \left( \frac{\dot{\sx}}{\dot{\st}_{\beta}}(s)\lambda^i_{\beta} \right) + \psi_{\beta}^* \Bigl( \nabla E(\sx(s),(\lambda^i_{\beta})^{-1}(\st_{\beta}(s)-t_{\beta}^i)+ t^i) \Bigr) \right]\frac{\dot{\st}_{\beta}(s)}{\lambda^i_{\beta}}\,ds.
\end{align*}
Applying Lemma~\ref{lemma:thm-limsup-estimates} in every subinterval  $[\ss(t^i),\ss(t^{i+1})]$, we obtain the same bounds (\ref{l-limsup-est:1}--\ref{l-limsup-est:3}) with $C_\beta S$ substituted by $C_\beta |\ss(t^{i+1}) - \ss(t^i)|$. Then inequality~\eqref{limsup-thm:ineq-to-prove} follows because
\[
\sum_i C_\beta  |\ss(t^{i+1}) - \ss(t^i)| = C_\beta S.
\]
The pointwise convergence of $x_\beta(t) \to x(t)$ for $t \in (\Sigma^i)^\circ$ is again trivial. The following calculations show that, by construction, the convergence holds also in the points $\{ t^i \} \supset J_x$
\[
x_{\beta}(t^i)\stackrel{\eqref{eq:def_gen_rec_seq}}{=} \sx \left( \st_{\beta}^{-1} \left( t^i_\beta \right) \right) = \sx \left( \st_{\beta}^{-1} \left( \st_\beta(\ss(t^i)) \right) \right) = \sx \left( \ss(t^i) \right) \stackrel{\eqref{eq:inverserelation}}{=} x(t^i) ,
\]
and we conclude.
\end{proof}

\begin{proof}[Proof of Lemma~\ref{lemma:thm-limsup-estimates}]
First note that for any $s'\in[0,S]$,
\[
0\leq\st_\beta(s')-\st(s') = \int_0^{s'} (\dot\st_\beta(s)-\dot \st(s))\, ds
\leq \int_0^S \e_\beta(s)\, ds 
\leq \int_0^S \frac{|\dot \sx(s)|}{\partial\psi_\beta^*(A+\delta_\beta)}\, ds
= \var(\sx,[0,T])K_\beta  = \frac{MK_{\beta}}{L}.
\]
Consequently
\begin{equation}
\label{ineq:lambda-to-zero}
0\leq T_\beta -T = T(\lambda_\beta-1) \leq \frac {MK_\beta}{L} \to 0 \qquad\text{as }\beta\to\infty.
\end{equation} 
Using the Lipschitz continuity of $\nabla E$ in time, we also have 
\begin{align}
\notag
|\nabla E(\sx(s),\st_{\beta}(s)\lambda_{\beta}^{-1})| - |\nabla E(\sx(s),\st(s))| 
&\leq L\big( \st_\beta(s)\lambda_\beta^{-1}-\st(s)\bigr)\\
&\leq L\big( \st_\beta(s)-\st(s)\bigr)  \leq MK_\beta.
\label{ineq:nablaE-Lipschitz}
\end{align}

We now  prove~\eqref{l-limsup-est:1}. We first estimate
\begin{align}
\notag
\int_0^S &\psi_\beta^*\Bigl(|\nabla E(\sx(s),\st_\beta(s)\lambda_\beta^{-1})|\Bigr) \frac{\dot \st_\beta(s)}{\lambda_\beta}\, ds
\stackrel{\eqref{ineq:nablaE-Lipschitz}}\leq \int_0^S \psi_\beta^*\Bigl(|\nabla E(\sx(s),\st(s))| + MK_\beta\Bigr) \frac{\dot \st_\beta(s)}{\lambda_\beta}\, ds\\
\label{ineq:lemma-limsup-estimates-psi*}
&\stackrel{\text{$\psi_\beta^*$ convex}}\leq \int_0^S 
 \left[\psi_\beta^*\Bigl(|\nabla E(\sx(s),\st(s))|\Bigr)
 + MK_\beta \partial\psi_\beta^*\Bigl(|\nabla E(\sx(s),\st(s))| + MK_\beta\Bigr)\right] \frac{\dot \st_\beta(s)}{\lambda_\beta}\, ds.
\end{align}
Setting $\Sigma_\beta := \{s\in[0,S]: \dot\st_\beta(s) = \dot\st(s)\}$, we split the domain into $\Sigma_\beta$ and $\Sigma_\beta^c$. On $\Sigma_\beta$, since $\dot\st(s) >0$, the finiteness of $\mathbb L$ implies that $|\nabla E|\leq A$; on $\Sigma_\beta^c$, $\dot \st_{\beta}(s) = \e_\beta(s)$. Therefore
\begin{multline}
\label{est:l-limsup-32a}
\int_0^S \psi_\beta^*\Bigl(|\nabla E(\sx(s),\st(s))|\Bigr) \bigl[\dot \st_\beta(s) - \e_\beta(s)\bigr]\, ds \\
=\int_{\Sigma_\beta} \psi_\beta^*\Bigl(|\nabla E(\sx(s),\st(s))|\Bigr) \bigl[\dot \st_\beta(s) - \e_\beta(s)\bigr]\, ds \\
\leq  \psi_\beta^*(A)\int_{\Sigma_\beta}  \dot \st_\beta(s)\, ds \leq C_\beta S,
\end{multline}
with $C_\beta:=\psi^*_\beta (A)$.

We also split the second term in~\eqref{ineq:lemma-limsup-estimates-psi*} into integrals on $\Sigma_\beta$ and $\Sigma_\beta^c$. On $\Sigma_\beta$, again since $|\nabla E| \leq A$
\begin{equation}
\label{est:l-limsup-32b}
\int_{\Sigma_\beta} MK_\beta \partial\psi_\beta^*\Bigl(|\nabla E(\sx(s),\st(s))| + MK_\beta\Bigr)\dot\st_\beta(s)\, ds \leq MK_\beta \psi_\beta^* (A+MK_\beta) \int_{\Sigma_\beta} \dot \st_\beta(s)\, ds \leq C_\beta S,
\end{equation}
with $C_\beta:=  MK_\beta \psi_\beta^* (A+MK_\beta)$.

On the other hand, on $\Sigma_\beta^c$, using $\dot \st_\beta(s) = \e_\beta(s) = |\dot \sx(s)| / \partial\psi_\beta^*\bigl(|\nabla E(\sx(s),\st(s))| \vee (A+\delta_\beta)\bigr)$
\begin{multline}
\int_{\Sigma_\beta^c} MK_\beta \partial\psi_\beta^*\Bigl(|\nabla E(\sx(s),\st(s))| + MK_\beta\Bigr)\dot\st_\beta(s)\, ds \leq \\
M\int_{\Sigma_\beta^c} K_\beta \frac{\partial\psi_\beta^*\Bigl(|\nabla E(\sx(s),\st(s))| + MK_\beta\Bigr)}{\partial\psi_\beta^*\Bigl(|\nabla E(\sx(s),\st(s))| \vee (A+\delta_\beta)\Bigr)}\, |\dot\sx(s)|\, ds \leq C_\beta S,
\label{est:l-limsup-32c}
\end{multline}
where the last inequality holds because $\var(\sx,[0,S]) \leq S$, with
\[
C_\beta:= \sup_s M K_\beta \frac{\partial\psi_\beta^*\Bigl(|\nabla E(\sx(s),\st(s))| + MK_\beta\Bigr)}{\partial\psi_\beta^*\Bigl(|\nabla E(\sx(s),\st(s))| \vee (A+\delta_\beta)\Bigr)}.
\]
Together, (\ref{est:l-limsup-32a}--\ref{est:l-limsup-32c}) prove~\eqref{l-limsup-est:1} with vanishing $C_\beta$ thanks to Hypothesis B-C.

\medskip

To prove~\eqref{l-limsup-est:2}, we use the fact that $\dot \st_\beta(s)\geq \e_\beta(s)$ for all $s$, and that the mapping $\tau \mapsto \tau \psi_\beta(v/\tau)$ is non-increasing. Therefore
\begin{align*}
\int_0^S \psi_{\beta} \biggl( \frac{\dot{\sx}}{\dot{\st}_{\beta}}(s)\lambda_{\beta} \biggr) \frac{\dot \st_\beta(s)}{\lambda_\beta} 
 \, ds 
&\leq \int_0^S \psi_{\beta} \biggl( \frac{\dot{\sx}}{\e_{\beta}}(s)\lambda_{\beta} \biggr) \frac{\e_\beta(s)}{\lambda_\beta} 
 \, ds \\
&\stackrel{(*)}\leq \int_0^S \biggl[\psi_{\beta} \biggl( \frac{\dot{\sx}}{\e_{\beta}}(s)\biggr) 
  + C \partial\psi_\beta^*\Bigl(|\nabla E(\sx(s),\st(s))|\vee (A+\delta_\beta) \Bigr)(\lambda_\beta-1)\biggr]{\e_\beta(s)} \, ds.
\end{align*}
where $C$ is a constant and we prove the inequality marked $(*)$ below.  Continuing with the argument, we again apply the definition~\eqref{def:epsilon} of $\e_\beta$ to find
\[
(\lambda_\beta-1) \int_0^S \partial\psi_\beta^*\Bigl(|\nabla E(\sx(s),\st(s))|\vee (A+\delta_\beta) \Bigr)\,{\e_\beta(s)} \, ds
\leq (\lambda_\beta-1) \int_0^S |\dot \sx(s)|\, ds \leq C_\beta S,
\]
where $C_\beta := (\lambda_\beta - 1)$ converges to zero by~\eqref{ineq:lambda-to-zero}. 
\medskip

We next prove~\eqref{l-limsup-est:3}, with $C_\beta:= \delta_\beta$. We calculate
\begin{align*}
\int_0^S  \biggl[ \psi_{\beta} \biggl( \frac{\dot{\sx}}{\e_\beta}(s) \biggr)
  &+ \psi_\beta^*\Bigl(|\nabla E(\sx(s),\st(s))|\Bigr)\biggr] \e_\beta(s) \, ds\\
&\leq \int_0^S  \biggl[ \psi_{\beta} \biggl( \frac{\dot{\sx}}{\e_\beta}(s) \biggr)
  + \psi_\beta^*\Bigl(|\nabla E(\sx(s),\st(s))| \vee(A+\delta_\beta)\Bigr)\biggr] \e_\beta(s) \, ds\\
&= \int_0^S |\dot\sx(s)| \Bigl(|\nabla E(\sx(s),\st(s))|\vee (A+\delta_\beta)\Bigr)\, ds
\qquad \text{by \eqref{eq:bipotential}}\\
&\leq  \int_0^S |\dot\sx(s)| \Bigl(|\nabla E(\sx(s),\st(s))|\vee A\Bigr)\, ds + C_\beta S.
\end{align*}

\medskip

We finally prove the inequality $(*)$ above, as the following separate result: for each $R>0$ there exists $C>0$ such that 
\[
\forall\, \alpha\geq 1,\; \forall \, |z|\leq R, \;\forall \,\beta:\qquad
\psi_\beta(\alpha\partial\psi_\beta^*(z)) \leq \psi_\beta(\partial\psi_\beta^*(z))
+ C(\alpha-1) \partial\psi_\beta^*(z).
\]
To show this, note that by Hypothesis D for each $\alpha\geq1$ and for each $|z|\leq R$ there exists $\eta_\beta(z,\alpha)$ such that 
\[
\partial\psi_\beta^*(z+\eta_\beta(z,\alpha)) = \alpha \partial\psi_\beta^*(z),
\]
and $\eta_\beta$ is bounded uniformly in $\alpha$, $\beta$, and $|z|\leq R$.
The following three statements follow from convexity and convex duality:
\begin{align*}
&\psi_\beta(\alpha\partial\psi_\beta^*(z)) + \psi_\beta^*(z+\eta_\beta) - (z+\eta_\beta)\alpha \partial\psi_\beta^*(z) = 0;\\
&\psi_\beta(\partial\psi_\beta^*(z)) + \psi_\beta^*(z) - z \partial\psi_\beta^*(z) = 0;\\
&\psi_\beta^*(z+\eta_\beta) - \psi_\beta^*(z) - \eta\partial\psi_\beta^*(z)\geq 0.
\end{align*}
Upon subtracting the second and third line from the first we find
\[
\psi_\beta(\alpha\partial\psi_\beta^*(z)) \leq \psi_\beta(\partial\psi_\beta^*(z))
+ (z+\eta_\beta)(\alpha-1) \partial\psi_\beta^*(z),
\]
which implies the result.
\end{proof}

\section{Discussion}
\label{sec:Discussion}


In the introduction we posed the question whether we could understand the distinction and the relationship between gradient-flow and rate-independent systems from the point of view of stochastic processes. The simple one-dimensional model of this paper gives a very clear answer, that we summarize in our words as follows:
\begin{itemize}
\item The continuum limit is a generalized gradient flow, with non-quadratic, non-$1$-homogeneous dissipation, and the large-deviations rate functional `is' the corresponding generalized gradient-flow structure, in the sense of Section~\ref{subsec:ldpgf-intro}; 
\item Taking further limits recovers both quadratic gradient-flow and rate-independent cases;
\item At least some of the limits are robust against exchanging the order of the limits, and we conjecture that this robustness goes much further. 
\end{itemize}
Therefore the quadratic and rate-independent cases are naturally embedded in the scale of systems characterized by $\alpha$ and $\beta$.

%
%
%

In addition, the details of the proofs show how the formulation in terms of $\J$ of (a) large deviations, (b) generalized gradient flows including rate-independent systems, and (c) convergence results for these systems, gives a unified view on the field and a coherent set of tools for the analysis and manipulation of the systems. 

\medskip

Related issues have been investigated in the case of stochastic differential equations.
The two limiting processes, $n\to\infty$ and $\beta\to\{0,\infty\}$ can be interpreted as differently scaled combinations of two limiting processes: (a) the small-noise limit, (b) the limit of vanishing microstructure. In the case of SDEs~\cite{Baldi91,FreidlinSowers99,DupuisSpiliopoulos12}, three regimes have been identified, corresponding to `microstructure smaller than noise', `noise smaller than microstructure', and the critical case. In the first of these, `microstructure smaller than noise', a behaviour arises that resembles the quadratic limit of this paper, in which the microstructure is effectively swamped by the noise. The critical case resembles our original large-deviation result (Theorem~\ref{thm:LDF}) in that both give non-quadratic, non-one-homogeneous rate functionals. Finally, when the noise is asymptotically smaller than the microstructure, a limit similar to the rate-independent limit is obtained in~\cite{FreidlinSowers99,DupuisSpiliopoulos12}, but because the authors  consider time-invariant energies and a different scaling, the behaviour of the limiting system is rather different.

\medskip

The one-dimensionality of the current setup may appear to be a significant restriction, but we believe (and in some cases we know) that the structure can be generalized to a wide class of other systems. For instance,
\begin{itemize}
\item The initial large-deviations result (Theorem~\ref{thm:LDF}) also holds in higher dimensions; other proofs of this and similar results are given in~\cite{ShwartzWeiss95,Chen96}.
\item The joint large-deviations-quadratic limit (Theorem~\ref{thm:lardevres}) generalizes to higher dimensions with only notational changes in the proof.
\item Of the proof of the convergence to a rate-independent system (Theorem~\ref{thm:RIgammaconv}), one part (the liminf-inequality) has been done in the generality of a metric space, with a specific functional form of the dissipation potential, in~\cite{MRS12}. The other part, the construction of a recovery sequence, is subject of current work; here the characterization of the limiting jump term depends on the particular form of the approximating $\psi_{\beta}$-$\psi_{\beta}^*$, in a way that is not yet clear. 
\end{itemize}
More generally, the results of~\cite{MielkePeletierRenger13TR} show that the connection between large-deviation principles and generalized gradient flows is robust, and arises for all reversible stochastic processes and quite a few more (such as the GENERIC system in~\cite{DuongPeletierZimmer13}). 

\medskip

In Figure~\ref{fig:ideaofpaper} the question mark represents an open problem: the combined large-deviation-rate-independent limit. We conjecture that, as in the combined large-deviation-quadratic limit (Theorem~\ref{thm:lardevres}), a large-devation principle holds in this limit, with rate functional $\J_{RI}$.  Unfortunately the framework provided by \cite{FK06} does not seem to apply as-is, and the form of this functional will require a radical change in the strategy of the proof. 


\section*{Acknowledgements}
We want to sincerely thank Jin Feng for the valuable comments and suggestions. During all the preparation of this work, Giuseppe Savar\'e provided many useful suggestions and critical remarks, and we are very grateful for his contribution. GAB and MAP kindly acknowledge support from the Nederlandse Organisatie voor Wetenschappelijk Onderzoek (NWO) VICI grant 639.033.008.


\bibliography{BonaschiBib}
\bibliographystyle{alpha}

\end{document}